\documentclass[11pt,letter]{article}
\usepackage[utf8]{inputenc}
\usepackage{geometry}
\usepackage{setspace,amsmath,amsfonts,amssymb,amsthm,tabularx,longtable,bbm,cite,bm,eucal,newfloat,
caption,wrapfig,fancyhdr,hyperref,caption,booktabs,graphicx,dsfont,csquotes,zref,lscape,algorithm, algorithmic,subcaption}
\usepackage[table,xcdraw]{xcolor}
\usepackage[square,sort,comma,longnamesfirst]{natbib}
\bibpunct{(}{)}{;}{a}{,}{,}

\definecolor{newred}{HTML}{E66100}
\definecolor{newgreen}{HTML}{44AA99}
\definecolor{newyellow}{HTML}{E899D5}
\makeatother
\hypersetup{colorlinks,linkcolor={newred},citecolor={newgreen},urlcolor={newyellow}} 
\DeclareMathOperator*{\E}{\mathbb{E}}
\newcommand{\argmin}{\arg\!\min}

\date{April 2021}
\usepackage[shortlabels]{enumitem}
\def\defeq{\equiv}
\newcommand{\EE}[2]{\mathbb{E}_{#1\!\!}\left[#2\right]}

\def\E#1{\EE{\,}{#1}}

\def\bE{{\mathbf E}}
\def\bI{{\mathbf I}}

\def\bA{{\mathbf A}}

\def\bf{{\mathbf f}}
\def\bg{{\mathbf g}}

\def\bu{{\mathbf u}}
\def\bv{{\mathbf v}}
\def\bw{{\mathbf w}}
\def\bB{{\mathbf B}}

\def\bU{{\mathbf U}}

\def\bD{{\mathbf D}}
\def\bT{{\mathbf T}}
\def\bS{{\mathbf S}}
\def\bW{{\mathbf W}}

\def\bC{{\mathbf C}}
\def\bx{{\mathbf x}}

\def\be{{\mathbf e}}
\def\bb{{\mathbf b}}
\def\by{{\mathbf y}}

\def\bs{{\mathbf s}}
\def\bUpsilon{{\bm \Upsilon}}
\def\bgamma{{\bm \gamma}}

\def\bLambda{{\bm \Lambda}}

\def\biota{{\bm \iota}}
\def\balpha{{\bm \alpha}}

\def\bxi{{\bm \xi}}

\def\bbeta{{\bm \beta}}

\def\btheta{{\bm \theta}}
\def\bzeta{{\bm \zeta}}

\def\bSigma{{\bm \Sigma}}
\def\bGamma{{\bm \Gamma}}
\def\bTheta{{\bm \Theta}}

\def\bvarepsilon{{\bm \varepsilon}}
\usepackage{mathtools}
\usepackage{wasysym}
\usepackage{indentfirst}
\DeclareMathAlphabet{\mathcal}{OMS}{cmsy}{m}{n}

\DeclarePairedDelimiter\abs{\lvert}{\rvert}%
\DeclarePairedDelimiter\norm{\lVert}{\rVert}%

\makeatletter
\let\oldabs\abs
\def\abs{\@ifstar{\oldabs}{\oldabs*}}
\let\oldnorm\norm
\def\norm{\@ifstar{\oldnorm}{\oldnorm*}}
\makeatother

\theoremstyle{plain}
\newtheorem{thm}{Theorem}

\newcommand{\vertiii}[1]{{\left\vert\kern-0.25ex\left\vert\kern-0.25ex\left\vert #1 
		\right\vert\kern-0.25ex\right\vert\kern-0.25ex\right\vert}}
{
	\theoremstyle{plain}
	
}

\makeatletter
\def\mathcolor#1#{\@mathcolor{#1}}
\def\@mathcolor#1#2#3{%
	\protect\leavevmode
	\begingroup
	\color#1{#2}#3%
	\endgroup
}
\makeatother

\numberwithin{equation}{section}

\usepackage[toc, page, title]{appendix}
\newtheorem{lem}{Lemma}
\begin{document}	
	\title{\textbf{Learning from Forecast Errors: \\A New Approach to Forecast Combinations}}
	\author{
		Tae-Hwy Lee\footnote{Department of Economics, University of California, Riverside. Email: tae.lee@ucr.edu.}\hskip 4mm \ and \hskip 2mm
		Ekaterina Seregina\footnote{Department of Economics, University of California, Riverside. Email: ekaterina.seregina@email.ucr.edu.}\hskip 8mm 
	}
	\maketitle
	\thispagestyle{empty}
	
	\begin{abstract}
		\begin{spacing}{2}
	Forecasters often use common information and hence make common mistakes. We propose a new approach, Factor Graphical Model (FGM), to forecast combinations that separates idiosyncratic forecast errors from the common errors. FGM exploits the factor structure of forecast errors and the sparsity of the precision matrix of the idiosyncratic errors. We prove the consistency of forecast combination weights and mean squared forecast error estimated using FGM, supporting the results with extensive simulations. Empirical applications to forecasting macroeconomic series shows that forecast combination using FGM outperforms combined forecasts using equal weights and graphical models without incorporating factor structure of forecast errors.
	\end{spacing}

		\vskip 2mm
		
		
		\newpage
	\end{abstract} 
	
	\newpage 
	\setlength{\baselineskip}{22pt}
	\setcounter{page}{1}
		\setstretch{1.9}
	\section{Introduction}
	A search for the best forecast combination has been an important  on-going research question in economics. \cite{CLEMEN1989} pointed out that combining forecasts is \enquote{practical, economical and useful. Many empirical tests have demonstrated the value of composite forecasting. We no longer need to justify that methodology}. However, as demonstrated by \cite{DIEBOLD2018}, there are still some unresolved issues. Despite the findings based on the theoretical grounds, equal-weighted forecasts have proved surprisingly difficult to beat. Many methodologies that seek for the best forecast combination use equal weights as a benchmark: for instance, \cite{DIEBOLD2018} develop \enquote{partially egalitarian Lasso}.
	
	The success of equal weights is partly due to the fact that the forecasters use the same set of public information to make forecasts, hence, they tend to make common mistakes. For example, in the European Central Bank's Survey of Professional forecasters of Euro-area real GDP growth, the forecasters tend to \textit{jointly} understate or overstate GDP growth. Therefore, we stipulate that the forecast errors include common and idiosyncratic components, which allows the forecast errors to move together due to the common error component. Our paper provides a simple framework to learn from analyzing forecast errors: we separate unique errors from the common errors to improve the accuracy of the combined forecast.

	Dating back to \cite{GrangerBatesWeights}, the well-known expression for the optimal forecast combination weights requires an estimator of inverse covariance (precision) matrix. Graphical models are a powerful tool to estimate precision matrix directly, avoiding the step of obtaining an estimator of covariance matrix to be inverted. Prominent examples of graphical models include Graphical Lasso (\cite{GLASSO}) and nodewise regression (\cite{meinshausen2006}). Despite using different strategies for estimating precision matrix, all graphical models assume that the latter is sparse: many entries of precision matrix are zero, which is a necessary condition to consistently estimate inverse covariance. Our paper demonstrates that such assumption contradicts the stylized fact that experts tend to make common mistakes and hence the forecast errors move together through common factors. We show that graphical models fail to recover entries of precision matrix under the factor structure.	
	
	This paper overcomes the aforementioned challenge and develops a new precision matrix estimator for the forecast errors under the approximate factor model with unobserved factors. We call our algorithm the \textit{Factor Graphical Model}. We use a factor model to estimate an idiosyncratic component of the forecast errors, and then apply a Graphical model (Graphical Lasso or nodewise regression) for the estimation of the precision matrix of the idiosyncratic component.
	
	There are a few papers that used graphical models in different contexts to estimate the covariance matrix of the idiosyncratic component when the factors are known and the loadings are assumed to be constant. \cite{Brownlees2018JAE} estimate a sparse covariance matrix for high-frequency data and construct the realized network for financial data. \cite{Brownlees2018EJS} develop a power-law partial correlation network based on the Gaussian graphical models. \cite{koike2019biased} uses the Weighted Graphical Lasso to estimate a sparse covariance matrix of the idiosyncratic component for a factor model with observable factors for high-frequency financial data.
	
	Our paper makes several contributions. First, we allow the forecast errors to be highly correlated  due to the common component which is motivated by the stylized fact that the forecasters tend to jointly understate or overstate the predicted series of interest. Second, we develop a high-dimensional precision matrix estimator which combines the benefits of the \textit{factor} structure and \textit{sparsity} of the precision matrix of the idiosyncratic component for the forecast combination under the approximate factor model. We prove consistency of forecast combination weights and the Mean Squared Forecast Error (MSFE) estimated using Factor Graphical models. Third, an empirical application to forecasting macroeconomic series in big data environment shows that incorporating the factor structure of the forecast errors into the graphical models improves the performance of a combined forecast over forecast combination using equal weights and graphical models without factors.
	
	The paper is structured as follows: Section 2 reviews Graphical Lasso and nodewise regression. Section 3 studies the approximate factor models for the forecast combination. Section 4 introduces the Factor Graphical Models and discusses the choice of the tuning parameters. Section 5 contains theoretical results and Section 6 validates these results
	using simulations. Section 7 studies an empirical application for macroeconomic time-series. Section 8 concludes and Section 9 collects the proofs of the theorems.
	
	\textbf{Notation}. For the convenience of the reader, we summarize the notation to be used throughout the paper. Let $\mathcal{S}_p$ denote the set of all $p \times p$ symmetric matrices. For any matrix $\bC$, its $(i,j)$-th element is denoted as $c_{ij}$. Given a vector $\bu\in \mathbb{R}^d$ and a parameter $a\in \lbrack1,\infty)$, let $\norm{\bu}_a$ denote $\ell_a$-norm. Given a matrix $\bU \in\mathcal{S}_p$, let $\Lambda_{\text{max}}(\bU) \defeq \Lambda_1(\bU) \geq \Lambda_2(\bU)\geq \ldots \geq \Lambda_{\text{min}}(\bU) \defeq \Lambda_p(\bU)$ be the eigenvalues of $\bU$.
	 Given a matrix $\bU \in \mathbb{R}^{p\times p}$ and parameters  $a,b\in \lbrack1,\infty)$, let $\vertiii{\bU}_{a,b}\defeq \max_{\norm{\by}_a=1}\norm{\bU\by}_{b}$ denote the induced matrix-operator norm. The special cases are $\vertiii{\bU}_1\defeq \max_{1\leq j\leq p}\sum_{i=1}^{p}\abs{u_{i,j}}$ for the $\ell_1/\ell_1$-operator norm;  the operator norm ($\ell_2$-matrix norm) $\vertiii{\bU}_{2}^{2}\defeq\Lambda_{\text{max}}(\bU\bU')$ is equal to the maximal singular value of $\bU$.
 Finally, $\norm{\bU}_{\infty}\defeq\max_{i,j}\abs{u_{i,j}}$ denotes the element-wise maximum.

	\section{Graphical Models for Forecast Errors}
	This section briefly reviews a class of models, called graphical models, that search for the estimator of the precision matrix. In graphical models, each vertex represents a random variable, and the graph visualizes the joint distribution of the entire set of random variables.
	\textit{Sparse graphs} have a relatively small number of edges.
	
	Suppose we have $p$ competing forecasts of the univariate series $y_t$,  $t=1,\ldots,T$. Let $\be_t=(e_{1t},\ldots,e_{pt})' \sim \mathcal{N} (\mathbf{0}, \bSigma)$ be a $p \times 1$ vector of forecast errors. Assume they follow a Gaussian distribution.
	The precision matrix $\bSigma^{-1}\defeq \bTheta$ contains information about partial covariances between the variables. For instance, if $\theta_{ij}$, which is the $ij$-th element of the precision matrix, is zero, then the variables $i$ and $j$ are conditionally independent, given the other variables.
	
	Let $\bW$ be the estimate of $\bSigma$. Given a sample $\{\be_t\}_{t=1}^{T}$, let $\bS = (1/T)\sum_{t=1}^{T}(\be_t)(\be_t)'$ denote the sample covariance matrix, which can be used as a choice for $\bW$. Also, let $\widehat{\bD}^2\defeq \textup{diag}(\bW)$. We can write down the Gaussian log-likelihood (up to constants) $l(\bTheta)=\log\det(\bTheta)-\text{trace}(\bW\bTheta)$. When $\bW=\bS$, the maximum likelihood estimator of $\bTheta$ is $\widehat{\bTheta}=\bS^{-1}$.
	
	In the high-dimensional settings it is necessary to regularize the precision matrix, which means that some edges will be zero. In the following subsections we discuss two most widely used techniques to estimate sparse high-dimensional precision matrices.
	\subsection{Graphical Lasso}
	The first approach to induce sparsity in the estimation of precision matrix is to add penalty to the maximum likelihood and use the connection between the precision matrix and regression coefficients to maximize the following \textit{weighted penalized log-likelihood} (\cite{Sara2018}):
	\begin{align} \label{eq62}
	&\widehat{\bTheta}_{\lambda}=\argmin_{\bTheta=\bTheta'}\textup{trace}(\bW\bTheta)-\log\det(\bTheta)+\lambda\sum_{i\neq j}\widehat{d}_{ii}\widehat{d}_{jj}\abs{\theta_{ij}},
	\end{align}
	over positive definite symmetric matrices, where $\lambda\geq0$ is a penalty parameter. The subscript $\lambda$ in $\widehat{\bTheta}_{\lambda}$ means that the solution of the optimization problem in \eqref{eq62} will depend upon the choice of the tuning parameter. More details on the latter are provided in Subsection 4.1 that describes how to choose the shrinkage intensity in practice. In order to simplify notation, we will omit the subscript.
	
	One of the most popular and fast algorithms to solve the optimization problem in \eqref{eq62} is called the Graphical Lasso (GLASSO), which was introduced by \cite{GLASSO}. Define the following partitions of $\bW$, $\bS$ and $\bTheta$:
	\begin{equation} \label{eq43}
	\bW=\begin{pmatrix}
	\underbrace{\bW_{11}}_{(p-1)\times(p-1)}&\underbrace{\bw_{12}}_{(p-1)\times 1}\\\bw_{12}'&w_{22}
	\end{pmatrix}, \bS=\begin{pmatrix}
	\underbrace{\bS_{11}}_{(p-1)\times(p-1)}&\underbrace{\bs_{12}}_{(p-1)\times 1}\\\bs_{12}'&s_{22}
	\end{pmatrix}, \bTheta=\begin{pmatrix}
	\underbrace{\bTheta_{11}}_{(p-1)\times(p-1)}&\underbrace{\btheta_{12}}_{(p-1)\times 1}\\\btheta_{12}'&\theta_{22}
	\end{pmatrix}.
	\end{equation}
	Let $\bbeta\defeq -\btheta_{12}/\theta_{22}$. The idea of GLASSO is to set $\bW= \bS+\lambda\bI$ in \eqref{eq62} and combine the gradient of \eqref{eq62} with the formula for partitioned inverses to obtain the following $\ell_1$-regularized quadratic program
	\begin{equation}\label{eq50}
	\widehat{\bbeta}=\argmin_{\bbeta\in \mathbb{R}^{p-1}}\Bigl\{ \frac{1}{2}\bbeta'\bW_{11}\bbeta-\bbeta'\bs_{12}+\lambda\norm{\bbeta}_1\Bigr\},
	\end{equation} 
	As shown by \cite{GLASSO}, \eqref{eq50} can be viewed as a LASSO regression, where the LASSO estimates are functions of the inner products of $\bW_{11}$ and $s_{12}$. Hence, \eqref{eq62} is equivalent to $p$ coupled LASSO problems. Once we obtain $\widehat{\bbeta}$, we can estimate the entries $\bTheta$ using the formula for partitioned inverses. GLASSO procedure is summarized in Algorithm \ref{alg1a}.
	\begin{spacing}{1.25}
	\begin{algorithm}[H]
		\caption{Graphical Lasso (\cite{GLASSO})}
		\label{alg1a}
		\begin{algorithmic}[1]
			\STATE 	Initialize $\bW= \bS+\lambda\bI$. The diagonal of $\bW$ remains the same in what follows.
			\STATE Repeat for $j=1,\ldots,p,1,\ldots,p,\ldots$ until convergence:
			\begin{itemize}
				\item Partition $\bW$ into part 1: all but the $j$-th row and column, and part 2: the $j$-th row and column.
				\item  Solve the score equations using the cyclical coordinate descent:
				\begin{equation*}
				\bW_{11}\bbeta-\bs_{12}+\lambda\cdot\text{Sign}(\bbeta)=\mathbf{0}.
				\end{equation*}
				This gives a $(p-1) \times 1$ vector solution $\widehat{\bbeta}.$
				\item Update $\widehat{\bw}_{12}=\bW_{11}\widehat{\bbeta}$.
			\end{itemize}
			\STATE In the final cycle (for $i=1,\ldots,p$) solve for 
			\begin{equation*}
			\frac{1}{\widehat{\theta}_{22}}=w_{22}-\widehat{\bbeta}'\widehat{\bw}_{12}, \quad 
			\widehat{\btheta}_{12}=-\widehat{\theta}_{22}\widehat{\bbeta}.
			\end{equation*}
		\end{algorithmic}
	\end{algorithm}	
\end{spacing}
As was shown in \cite{GLASSO}, the estimator produced by Algorithm \ref{alg1a} is guaranteed to be positive definite. Furthermore, \cite{Sara2018} showed that Algorithm \ref{alg1a} is guaranteed to converge and produces consistent estimator of precision matrix under certain sparsity conditions. 
\subsection{Nodewise Regression}
An alternative approach to induce sparsity in the estimation of precision matrix in equation \eqref{eq62} is to solve for $\widehat{\bTheta}$ one column at a time via linear regressions, replacing population moments by their sample counterparts $\bS$. When we repeat this procedure for each variable $j=1,\ldots,p$, we will estimate the elements of $\widehat{\bTheta}$ column by column using $\{\be_t\}_{t=1}^{T}$ via $p$ linear regressions. \cite{meinshausen2006} use this approach (which we will refer to as MB) to incorporate sparsity into the estimation of the precision matrix. Instead of running $p$ coupled LASSO problems as in GLASSO, they fit $p$ separate LASSO regressions using each variable (node) as the response and the others as predictors to estimate $\widehat{\bTheta}$. This method is known as the \enquote{nodewise} regression and it is reviewed below based on \cite{Buhlmann2014} and \cite{Caner2019}.

Let $\be_j$ be a $T \times 1$ vector of observations for the $j$-th regressor, the remaining covariates are collected in a $T \times p$ matrix $\bE_{-j}$.  For each $j=1,\ldots,p$ we run the following Lasso regressions:
\begin{equation}\label{e21}
\widehat{\bgamma}_j= \argmin_{\bgamma\in \mathbb{R}^{p-1}}\Big(\norm{\be_j-\bE_{-j}\bgamma}_{2}^{2}/T+2\lambda_j\norm{\bgamma}_1\Big),
\end{equation}
where $\widehat{\bgamma}_j=\{\widehat{\gamma}_{j,k}; j=1,\ldots,p, k\neq j\}$ is a $(p-1)\times 1$ vector of the estimated regression coefficients that will be used to construct the estimate of the precision matrix, $\widehat{\bTheta}$. Define
\begin{equation}
\widehat{\bC}=\begin{pmatrix}
1&-\widehat{\gamma}_{1,2}&\cdots&-\widehat{\gamma}_{1,p}\\
-\widehat{\gamma}_{2,1}&1&\cdots&-\widehat{\gamma}_{2,p}\\
\vdots&\vdots&\ddots&\vdots\\
-\widehat{\gamma}_{p,1}&-\widehat{\gamma}_{p,2}&\cdots&1\\
\end{pmatrix}.
\end{equation} 
For $j=1,\ldots,p$, define
\begin{equation}
\hat{\tau}_{j}^{2} = \norm{\be_j-\bE_{-j}\widehat{\bgamma}_j}_{2}^{2}/T+\lambda_j\norm{\widehat{\bgamma}_j}_1
\end{equation}
and write
\begin{equation}
\widehat{\bT}^2 = \text{diag}(\hat{\tau}_{1}^{2},\ldots,\hat{\tau}_{p}^{2}).
\end{equation}
The approximate inverse is defined as
\begin{equation} \label{e25}
\widehat{\bTheta}_{\lambda_j} = \widehat{\bT}^{-2} \widehat{\bC}.
\end{equation}
Similarly to GLASSO, the subscript $\lambda_j$ in $\widehat{\bTheta}_{\lambda_j}$ means that the estimated $\bTheta$ will depend upon the choice of the tuning parameter: more details are provided in Subsection 4.1 which discusses how to choose shrinkage intensity in practice. The subscript is omitted to simplify the notation. The procedure to estimate the precision matrix using nodewise regression is summarized in Algorithm \ref{alg1b}.
\begin{spacing}{1.25}
\begin{algorithm}[H]
	\caption{Nodewise regression by \cite{meinshausen2006} (MB)}
	\label{alg1b}
	\begin{algorithmic}[1]
		\STATE Repeat for $j=1,\ldots,p$ :
		\begin{itemize}
			\item Estimate $\widehat{\bgamma}_j$ using \eqref{e21} for a given $\lambda_j$.
			\item  Select $\lambda_j$ using a suitable information criterion (see section 4.1 for the possible options).
		\end{itemize}
		\STATE Calculate $\widehat{\bC}$ and $\widehat{\bT}^2$ .
		\STATE Return $\widehat{\bTheta} = \widehat{\bT}^{-2}  \widehat{\bC}$.
	\end{algorithmic}
\end{algorithm}	
\end{spacing}
One of the caveats to keep in mind when using the MB method is that the estimator in \eqref{e25} is not self-adjoint. \cite{Caner2019} show (see their Lemma A.1) that $\widehat{\bTheta}$ in \eqref{e25} is positive definite with high probability, however, it could still occur that $\widehat{\bTheta}$ is not positive definite in finite samples. In such cases we use the matrix symmetrization procedure as in \cite{fan2018elliptical} and then use eigenvalue cleaning as in \cite{Callot2017} and \cite{Hautsch2012}.
	\section{Approximate Factor Models for Forecast Errors}
	The approximate factor models for the forecasts were first considered by \cite{StockChan1999}. They modeled a panel of ex-ante forecasts of a single time-series as a dynamic factor model and found out that the combined forecasts improved on individual ones when all forecasts have the same information set (up to difference in lags). This result emphasizes the benefit of forecast combination even when the individual forecasts are not based on different information and, therefore, do not broaden the information set used by any one forecaster.
	
	In this paper, we are interested in finding the combination of forecasts which yields the best out-of-sample performance in terms of the mean-squared forecast error. We claim that the forecasters use the same set of public information to make forecasts and hence they tend to make common mistakes. Figure \ref{fig1} illustrates this statement: it shows quarterly forecasts of Euro-area real GDP growth produced by the European Central Bank's Survey of Professional Forecasters from 1999Q3 to 2019Q3. As described in \cite{DIEBOLD2018}, forecasts are solicited for one year ahead of the latest available outcome: e.g., the 2007Q1 survey asked the respondents to forecast the GDP growth over 2006Q3-2007Q3. As evidenced from Figure \ref{fig1}, forecasters tend to jointly understate or overstate GDP growth, meaning that their forecast errors include common and idiosyncratic parts. Therefore, we can model the tendency of the forecast errors to move together via factor decomposition.

	Recall that we have $p$ competing forecasts of the univariate series $y_t$,  $t=1,\ldots,T$ and $\be_t=(e_{1t},\ldots,e_{pt})' \sim \mathcal{N} (\mathbf{0}, \bSigma)$ is a $p \times 1$ vector of forecast errors. Assume that the generating process for the forecast errors follows a $q$-factor model:
	\begin{align} \label{equ1}
	&\underbrace{\be_t}_{p \times 1}=\bB \underbrace{\bf_t}_{q\times 1}+ \ \bvarepsilon_t,\quad t=1,\ldots,T
	\end{align}
	where $\bf_t=(f_{1t},\ldots, f_{qt})'$ are the common factors of the forecast errors for $p$ models, $\bB$ is a $p \times q$ matrix of factor loadings, and $\bvarepsilon_t$ is the idiosyncratic component that cannot be explained by the common factors. Unobservable factors, $\bf_{t}$, and loadings, $\bB$, are usually estimated by the principal component analysis (PCA), studied in \cite{Connor1988,Bai2003,Bai2002,Stock2002}. Strict factor structure assumes that the idiosyncratic forecast error terms, $\bvarepsilon_{t}$, are uncorrelated with each other, whereas approximate factor structure allows correlation of the idiosyncratic components (\cite{Chamberlain}).

	We use the following notations: $\E{\bvarepsilon_t\bvarepsilon'_t}=\bSigma_{\varepsilon}$, $\E{\bf_t\bf'_t}=\bSigma_{f}$, $\E{\be_t\be'_t}=\bSigma=\bB\bSigma_{f}\bB'+ \bSigma_{\varepsilon}$, and $\E{\bvarepsilon_t|\bf_{t}}=0$.  Let $\bTheta=\bSigma^{-1}$, $\bTheta_{\varepsilon}=\bSigma_{\varepsilon}^{-1}$ and $\bTheta_{f}=\bSigma_{f}^{-1}$ be the precision matrices of forecast errors, idiosyncratic and common components respectively.
	The objective function to recover factors and loadings from \eqref{equ1} is:
	\begin{align} \label{eq7}
	&\min_{\bf_1,\ldots, \bf_T, \bB}\frac{1}{T}\sum_{t=1}^{T}(\be_t-\bB\bf_t)'(\be_t-\bB\bf_t)\\
	&\text{s.t.} \ \bB'\bB=\bI_q, \label{eq9}
	\end{align}
	where \eqref{eq9} is the assumption necessary for the unique identification of factors. Fixing the value of $\bB$, we can project forecast errors $\be_t$ into the space spanned by $\bB$: $\bf_t=(\bB'\bB)^{-1}\bB'\be_t=\bB'\be_t$. When combined with \eqref{eq7}, this yields a concentrated objective function for $\bB$:
	\begin{equation} \label{equ10}
	\max_{\bB} \ \text{tr}\Big[\bB'\Big(\frac{1}{T}\sum_{t=1}^{T}\be_t\be_t'\Big) \bB\Big].
	\end{equation}
	It is well-known (see \cite{Stock2002} among others) that $\widehat{\bB}$ estimated from the first $q$ eigenvectors of $\frac{1}{T}\sum_{t=1}^{T}\be_t\be_t'$ is the solution to \eqref{equ10}.
	Given a sample of the estimated residuals $\{\widehat{\bvarepsilon}_t=\be_t-\widehat{\bB}\widehat{\bf_t}\}_{t=1}^{T}$ and the estimated factors $\{\widehat{\bf}_t\}_{t=1}^{T}$, let $\widehat{\bSigma}_{\varepsilon} = (1/T)\sum_{t=1}^{T}\widehat{\bvarepsilon}_t\widehat{\bvarepsilon}_t'$ and $\widehat{\bSigma}_{f}=(1/T)\sum_{t=1}^{T}\widehat{\bf}_t\widehat{\bf}_t'$ be the sample counterparts of the covariance matrices.
	
	Moving forward to the forecast combination exercise, suppose we have $p$ competing forecasts, $\widehat{\by}_{t}=(\hat{y}_{1,t},\ldots,\hat{y}_{p,t})'$, of the variable $y_t$, $t=1,\ldots,T$. The forecast combination is defined as follows:
	\begin{align}
	\widehat{y}_{t}^{c}=\bw'\widehat{\by}_{t} \label{e14}
	\end{align}
	where $\bw$ is a $p \times 1$ vector of weights. Define a measure of risk $\text{MSFE}(\bw, \bSigma)=\bw'\bSigma\bw$. As shown in \cite{GrangerBatesWeights}, the \textit{optimal} forecast combination minimizes the variance of the combined forecast error:
	\begin{equation} \label{equ11}
	 \min_{\bw} \text{MSFE}=\min_{\bw}\E{\bw'\be_t\be^{'}_{t}\bw}=\min_{\bw} \bw'\bSigma\bw, \ \text{s.t.} \ \bw'\biota_p =1,
	\end{equation}
	where $\biota_p$ is a $p\times 1$ vector of ones. The solution to \eqref{equ11} yields a $p\times 1$ vector of the optimal forecast combination weights:
	\begin{equation} \label{eq13}
	\bw=\frac{\bTheta\biota_p}{\biota_p'\bTheta\biota_p}.
	\end{equation}
 	If the true precision matrix is known, the equation \eqref{eq13} guarantees to yield the optimal forecast combination. In reality, one has to estimate $\bTheta$. Hence, the out-of-sample performance of the combined forecast is affected by the estimation error. As pointed out by \cite{smith2009simple}, when the estimation uncertainty of the weights is taken into account, there is no guarantee that the \enquote{optimal} forecast combination will be better than the equal weights or even improve the individual forecasts. Define $a = \biota'_{p}\bTheta\biota_p/p$, and $\widehat{a} = \biota'_{p}\widehat{\bTheta}\biota_p/p$. We can write
	\begin{align}\label{e4.8}
	\abs{\frac{\text{MSFE}(\widehat{\bw},\widehat{\bSigma})}{\text{MSFE}(\bw,\bSigma)} -1} = \abs{ \frac{\hat{a}^{-1} }{a^{-1}}-1}=\frac{\abs{a-\hat{a}}}{\abs{\hat{a}}},
	\end{align}
	and
	\begin{align} \label{e4.9}
	\norm{\widehat{\bw}-\bw}_1\leq \frac{a\frac{\norm{(\widehat{\bTheta}-\bTheta)\biota_p}_1 }{p}+\abs{a-\widehat{a} }\frac{\norm{\bTheta\biota_p}_1}{p}}{\abs{\widehat{a}}a}.
	\end{align}
	Therefore, in order to control the estimation uncertainty in the MSFE and combination weights, one needs to obtain a consistent estimator of the precision matrix $\bTheta$. More details are discussed in Subsection 5.2 and Theorems \ref{theor1} and \ref{theor2}.
	
	\section{Factor Graphical Models for Forecast Errors}	
	Since our interest is in constructing weights for the forecast combination, our goal is to estimate a precision matrix of the forecast errors. However, as pointed out by \cite{koike2019biased}, when common factors are present across the forecast errors, the precision matrix cannot be sparse because all pairs of the forecast errors are partially correlated given other forecast errors through the common factors. To illustrate this point, we generated forecast errors that follow \eqref{equ1} with $q=2$ and $\varepsilon_t\sim \mathcal{N}(\mathbf{0},\bSigma_{\varepsilon})$, where $\sigma_{\varepsilon,ij}=0.4^{\abs{i-j}}$ is the $i,j$-th element of $\bSigma_{\varepsilon}$. The vector of factors $\bf_{t}$ is drawn from $\mathcal{N}(\mathbf{0},\bI_q/10)$, and the entries of the matrix of factor loadings for forecast error $j=1,\ldots,p$, $\bb_{j}$, are drawn from $\mathcal{N}(\mathbf{0},\bI_q/100)$. The full loading matrix is given by $\bB=(\bb_{1},\ldots,\bb_{p})'$. Let $\widehat{q}$ denote the number of factors estimated by the PCA. We set $(T,p)=(1000,50)$ and plot the heatmap and histogram of population partial correlations of forecast errors $\be_{t}$, which are the entries of a precision matrix, in \autoref{fig1a}. We now examine the performance of graphical models for estimating partial correlations under the factor structure. \autoref{fig1aa} shows the partial correlations estimated by GLASSO that does not take into account factors: due to strict sparsity imposed by graphical models almost all partial correlations are shrunk to zero which degenerates the histogram in \autoref{fig1aa}. This means that strong sparsity assumption on $\bTheta$ imposed by classical graphical models (such as GLASSO and nodewise regression from Algorithms \ref{alg1a}-\ref{alg1b}) is not realistic under the factor structure.

	 In order to avoid the aforementioned problem, instead of imposing sparsity assumption on the precision of forecast errors, $\bTheta$, we require sparsity of the precision matrix of the idiosyncratic errors, $\bTheta_{\varepsilon}$. The latter is obtained using the estimated residuals after removing the co-movements induced by the factors (see \cite{Brownlees2018EJS,Brownlees2018JAE,koike2019biased}). Naturally, once we condition on the common components, it is sensible to assume that many remaining partial correlations of $\bvarepsilon_{t}$ will be negligible and thus $\bTheta_{\varepsilon}$ is sparse.

	We use the weighted Graphical Lasso and nodewise regression as shrinkage techniques to estimate the precision matrix of residuals. Once the precision of the low-rank component is obtained, we use the Sherman-Morrison-Woodbury formula to estimate the precision of forecast errors:
	\begin{equation}\label{equa18}
	\bTheta=\bTheta_{\varepsilon}-\bTheta_{\varepsilon}\bB\lbrack\bTheta_f+\bB'\bTheta_{\varepsilon}\bB\rbrack^{-1}\bB'\bTheta_{\varepsilon}.
	\end{equation}
	To obtain $\widehat{\bTheta}_{f}=\widehat{\bSigma}_{f}^{-1}$, we use $\widehat{\bSigma}_{f}=\frac{1}{T}\sum_{t=1}^{T}\widehat{\bf}_{t}\widehat{\bf}_{t}^{'}$. To get $\widehat{\bTheta}_{\varepsilon}$, we develop two approaches: the first uses the weighted GLASSO Algorithm \ref{alg1a}, with the initial estimate of the covariance matrix of the idiosyncratic errors calculated as $\widehat{\bSigma}_{\varepsilon}=\frac{1}{T}\sum_{t=1}^{T}\widehat{\bvarepsilon}_{t}\widehat{\bvarepsilon}_{t}^{'}$, where $\widehat{\bvarepsilon}_{t}=\be_t-\widehat{\bB}\widehat{\bf}_{t}$. The second uses nodewise regression and applies Algorithm \ref{alg1b} to $\widehat{\bvarepsilon}_t$. Once we estimate $\widehat{\bTheta}_{f}$ and $\widehat{\bTheta}_{\varepsilon}$, we can get $\widehat{\bTheta}$ using a sample analogue of \eqref{equa18}.
	We call the proposed procedures \textit{Factor Graphical Lasso} and \textit{Factor nodewise regression}  and summarize them in Algorithm \ref{alg2} and Algorithm \ref{alg3} respectively.	
	\begin{spacing}{1.25}
		\begin{algorithm}[H]
		\caption{Factor Graphical Lasso (Factor GLASSO)}
		\label{alg2}
		\begin{algorithmic}[1]
			\STATE Estimate factors, $\widehat{\bf}_t$, and factor loadings, $\widehat{\bB}$, using PCA. Obtain $\widehat{\bSigma}_{f}=\frac{1}{T}\sum_{t=1}^{T}\widehat{\bf}_{t}\widehat{\bf}_{t}^{'}$, $\widehat{\bTheta}_{f}=\widehat{\bSigma}_{f}^{-1}$, $\widehat{\bvarepsilon}_t = \be_t-\widehat{\bB}\widehat{\bf_t}$, and $\widehat{\bSigma}_{\varepsilon}=\frac{1}{T}\sum_{t=1}^{T}\widehat{\bvarepsilon}_{t}\widehat{\bvarepsilon}_{t}^{'}$.
	
			\STATE Estimate a sparse $\bTheta_{\varepsilon}$ using the weighted Graphical Lasso in \eqref{eq62} initialized with $\bW_{\varepsilon}=\widehat{\bSigma}_{\varepsilon}+\lambda\bI$:
			\begin{align} \label{e7.6}
			&\widehat{\bTheta}_{\varepsilon,\lambda}=\argmin_{\bTheta_{\varepsilon}=\bTheta'_{\varepsilon}}\text{trace}(\bW_{\varepsilon}\bTheta_{\varepsilon})-\log\det(\bTheta_{\varepsilon})+\lambda\sum_{i\neq j}\widehat{d}_{\varepsilon,ii}\widehat{d}_{\varepsilon,jj}\abs{\theta_{\varepsilon,ij}}.
			\end{align}
			to get $\widehat{\bTheta}_{\varepsilon}$.
			
			\STATE Use $\widehat{\bTheta}_f$ from Step 1 and $\widehat{\bTheta}_{\varepsilon}$ from Step 2 to estimate $\bTheta$ using the sample counterpart of the Sherman-Morrison-Woodbury formula in \eqref{equa18}:
			\begin{equation}\label{3.11}
			\widehat{\bTheta}=\widehat{\bTheta}_{\varepsilon}-\widehat{\bTheta}_{\varepsilon}\widehat{\bB}\lbrack\widehat{\bTheta}_f+\widehat{\bB}'\widehat{\bTheta}_{\varepsilon}\widehat{\bB}\rbrack^{-1}\widehat{\bB}'\widehat{\bTheta}_{\varepsilon}.
			\end{equation}

		\end{algorithmic}
	\end{algorithm}	
\end{spacing}
	\begin{spacing}{1.25}
		\begin{algorithm}[H]
	\caption{Factor nodewise regression \cite{meinshausen2006} (Factor MB)}
	\label{alg3}
	\begin{algorithmic}[1]
		\STATE Estimate factors, $\widehat{\bf}_t$, and factor loadings, $\widehat{\bB}$, using PCA. Obtain $\widehat{\bSigma}_{f}=\frac{1}{T}\sum_{t=1}^{T}\widehat{\bf}_{t}\widehat{\bf}_{t}^{'}$, $\widehat{\bTheta}_{f}=\widehat{\bSigma}_{f}^{-1}$, and $\widehat{\bvarepsilon}_t = \be_t-\widehat{\bB}\widehat{\bf_t}$.
		
		\STATE Estimate a sparse $\bTheta_{\varepsilon}$ using nodewise regression: let $\widehat{\bvarepsilon}_j$ be a $T\times 1$ vector of observations for the $j$-th regressor, and $\widehat{\bUpsilon}_{-j}$ is a $T \times p$ matrix that collects the remaining covariates. Run LASSO regressions in \eqref{e21} for $\widehat{\bvarepsilon}_t$:
			\begin{equation}\label{e21a}
		\widehat{\bgamma}_{\varepsilon,j}= \argmin_{\bgamma_{\varepsilon}\in \mathbb{R}^{p-1}}\Big(\norm{\widehat{\bvarepsilon}_j-\widehat{\bUpsilon}_{-j}\bgamma_{\varepsilon}}_{2}^{2}/T+2\lambda_j\norm{\bgamma_{\varepsilon}}_1\Big),
		\end{equation}
		to get $\widehat{\bTheta}_{\varepsilon}$.
		
	\STATE Use $\widehat{\bTheta}_f$ from Step 1 and $\widehat{\bTheta}_{\varepsilon}$ from Step 2 to estimate $\bTheta$ using the sample counterpart of the Sherman-Morrison-Woodbury formula in \eqref{equa18}:
	\begin{equation}\label{3.11a}
	\widehat{\bTheta}=\widehat{\bTheta}_{\varepsilon}-\widehat{\bTheta}_{\varepsilon}\widehat{\bB}\lbrack\widehat{\bTheta}_f+\widehat{\bB}'\widehat{\bTheta}_{\varepsilon}\widehat{\bB}\rbrack^{-1}\widehat{\bB}'\widehat{\bTheta}_{\varepsilon}.
	\end{equation}
	\end{algorithmic}
\end{algorithm}	
\end{spacing}
	Note that Algorithms \ref{alg2} and \ref{alg3} involve the tuning parameters $\lambda$ and $\lambda_j$, the procedure on how to choose the
shrinkage intensity coefficients is described in more detail in Subsection 4.1 that describes how to choose the shrinkage intensity in practice, and Section 5 that establishes sparsity requirements that guarantee convergence of \eqref{e7.6}, \eqref{3.11}, \eqref{e21a}, and \eqref{3.11a}.

	We can use $\widehat{\bTheta}$ to estimate the forecast combination weights $\widehat{\bw}$
	\begin{equation} \label{eq35}
	\widehat{\bw}=\frac{\widehat{\bTheta}\biota_p}{\biota_p'\widehat{\bTheta}\biota_p},
	\end{equation}
	where $\widehat{\bTheta}$ is obtained from Algorithm \ref{alg2} or Algorithm \ref{alg3}. Let us now revisit the motivating example at the beginning of this section: Figures \ref{fig2a}-\ref{fig4a} plot the heatmaps and the estimated partial correlations when precision matrix is computed using Factor GLASSO in Algorithm \autoref{alg2} with $\widehat{q}\in\{1,2,3\}$ statistical factors. The heatmaps and histograms closely resemble population counterparts in Figure \ref{fig1a}, and the result is not very sensitive to over- or under-estimating the number of factors $\widehat{q}$. This demonstrates that using a combination of classical graphical models and factor structure via Factor Graphical Models in Algorithms \ref{alg2}-\ref{alg3} improves upon the performance of classical graphical models: our approach allows to extract the benefits of modeling common movements in forecast errors, captured by a factor model, and the benefits of using many competing forecasting models that give rise to a high-dimensional precision matrix, captured by a graphical model.
	\subsection{The Choice of the Tuning Parameters for FGM}
	Algorithms \ref{alg2}-\ref{alg3} require the tuning parameters $\lambda$ (from Algorithm \ref{alg1a}) and $\lambda_j$ (from Algorithm \ref{alg1b}) respectively. We now comment on the choices for both tuning parameters.
	
	To motivate the choice of the tuning parameter for GLASSO and Factor GLASSO, we first briefly discuss some of the existing options to motivate our choice of $\lambda$ in \eqref{eq62} in simulations and the empirical application. Usually $\lambda$ is selected from a grid of values $F_{\lambda}=(\lambda_{\text{min}},\ldots,\lambda_{\text{max}})$ which minimizes the score measuring the goodness-of-fit. Some popular examples include  multifold cross-validation (CV), Stability Approach to Regularization Selection (STARS, \cite{STARS}), and the Extended Bayesian Information Criteria (EBIC, \cite{EBIC}). Since we are interested in estimating a sparse high-dimensional precision matrix, we need to choose a method for selecting the tuning parameter which is consistent in high-dimensions. \cite{MeinshausenCV} suggest that CV performs poorly for high-dimensional data, it overfits (\cite{STARS}), and it does not consistently select models. \cite{Zhu_STARS} pointed out that the STARS is not computationally efficient. It is consistent under certain conditions, but suffers from the problem of overselection in estimating Gaussian graphical models. In contrast, EBIC is computationally efficient and is considered to be the state-of-the-art technique for choosing the tuning parameter for the undirected graphs. The score measuring the goodness of fit for EBIC can be written as:
	\begin{align}\label{e109}
	\lambda_{\text{EBIC}}=\argmin_{\lambda \in F_{\lambda}}\{-2l(\bTheta_{\varepsilon,\lambda})+\log(T)\text{df}(\bTheta_{\varepsilon,\lambda})+4\text{df}(\bTheta_{\varepsilon,\lambda})\log(p)\eta	\},
	\end{align}
	where $\eta\in \lbrack0,1\rbrack$, $\bTheta_{\varepsilon,\lambda}$ is the precision matrix estimated for the tuning parameter $\lambda\in F_{\lambda}$, and the log-likelihood is  $l(\bTheta_{\varepsilon,\lambda})=\log\det(\bTheta_{\varepsilon,\lambda})-\text{trace}(\bW_{\varepsilon}\bTheta_{\varepsilon})$. For the estimation of graphical models, the degrees of freedom are usually defined as the number of unique non-zero elements in the estimated precision matrix, $\text{df}(\bTheta_{\varepsilon,\lambda})=\sum_{i\leq j} I_{\bTheta_{\varepsilon,\lambda,i,j}\neq 0}$. \cite{Chen_EBIC} showed that when $\eta=1$, EBIC is consistent as long as the dimension $p$ does not grow exponentially with the sample size $T$. Hence, in our simulations and the empirical exercise we use EBIC with $\eta=1$ for GLASSO and Factor GLASSO in Algorithms \ref{alg1a} and \ref{alg2}.
	
	For Algorithms \ref{alg1b} and \ref{alg3}, we follow \cite{Caner2019} to choose $\lambda_j$ in \eqref{e21} by minimizing the generalized information criterion (GIC). Let $\abs{\widehat{S}_j(\lambda_j)}$ denote the estimated number of nonzero parameters in the vector $\widehat{\bgamma}_{\varepsilon,j}$:
	\begin{equation}
	\text{GIC}(\lambda_j) = \log\Big(\norm{\widehat{\bvarepsilon}_j-\widehat{\bUpsilon}_{-j}\bgamma_{\varepsilon}}_{2}^{2}/T\Big)+\abs{\widehat{S}_j(\lambda_j)}\frac{\log(p)}{T}\log(\log(T)).
	\end{equation}
	As pointed out by \cite{Caner2019}, the GIC selects the true model with probability approaching one both when $p>T$ and when $p\leq T$.
	\section{Asymptotic Properties}
	We first introduce some terminology and notations. Let $A\in \mathcal{S}_p$. Define the following set for $j=1,\ldots,p$:
	\begin{align}\label{equ84}
	&D_j(A)\defeq\{ i:A_{ij}\neq 0,\ i\neq j\}, \quad d_j(A)\defeq\text{card}(D_j(A)),\quad d(A)\defeq\max_{j=1,\ldots,p}d_j(A),
	\end{align}
	where $d_j(A)$ is the number of edges adjacent to the vertex $j$ (i.e., the \textit{degree} of vertex $j$), and $d(A)$ measures the maximum vertex degree. Define $S(A)\defeq \bigcup_{j=1}^{p}D_j(A)$ to be the overall off-diagonal sparsity pattern, and $s(A)\defeq \sum_{j=1}^{p}d_j(A)$ is the overall number of edges contained in the graph. Note that $\text{card}(S(A)) \leq s(A)$: when $s(A)=p(p-1)/2$ this would give a fully connected graph.
	
	For the nodewise regression in \eqref{e21a}, denote $D_j\defeq \{ k; \gamma_{j,k}\neq 0 \}$ to be the active set for row $\bgamma_{j}$, and let $d_j\defeq \abs{D_j}$. Define $\bar{d}\defeq \max_{1\leq j\leq p}d_j$.
	\subsection{Assumptions}
	 We now list the assumptions on the model \eqref{equ1}:
	\begin{enumerate}[\textbf{({A}.1)}]
		\item \label{A1} (Spiked covariance model)
		As $p \rightarrow \infty$, $\Lambda_1(\bSigma)>\Lambda_2(\bSigma)+\ldots>\Lambda_q(\bSigma)\gg \Lambda_{q+1}(\bSigma)\geq \ldots \geq \Lambda_p(\bSigma) > 0$, where $\Lambda_j(\bSigma)=\mathcal{O}(p)$ for $j \leq q$, while the non-spiked eigenvalues are bounded, $\Lambda_j(\bSigma)=o(p)$ for $j > q$. We further require that $\Lambda_1(\bSigma)$ is uniformly bounded away from infinity.
	\end{enumerate}
	\begin{enumerate}[\textbf{({A}.2)}]
		\item \label{A2}(Pervasive factors)
		There exists a positive definite $q \times q$ matrix $\breve{\bB}$ such that\\ $\vertiii{p^{-1}\bB'\bB-\breve{\bB}}_{2}\rightarrow 0$ and $\Lambda_{\text{min}}(\breve{\bB})^{-1}=\mathcal{O}(1)$  as $p \rightarrow \infty$.
	\end{enumerate}
	We also impose strong mixing condition. Let $\mathcal{F}_{-\infty}^{0}$ and $\mathcal{F}_{T}^{\infty}$ denote the $\sigma$-algebras that are generated by $\{(\bf_t,\bvarepsilon_{t}):t\leq 0\}$ and $\{(\bf_t,\bvarepsilon_{t}):t\geq T\}$ respectively. Define the mixing coefficient
	\begin{equation}
	\alpha(T)=\sup_{A\in \mathcal{F}_{-\infty}^{0}, B \in \mathcal{F}_{T}^{\infty}}\abs{\Pr{A}\Pr{B}-\Pr{AB}}.
	\end{equation}
	\begin{enumerate}[\textbf{({A}.3)}]
	\item \label{A3} (Strong mixing) There exists $r_3>0$ such that $3r_{1}^{-1}+1.5r_{2}^{-1}+3r_{3}^{-1}>1$, and $C>0$ satisfying, for all $T\in \mathbb{Z}^{+}$, $\alpha(T)\leq \exp (-CT^{r_3})$.
	\end{enumerate}	
	Assumption \ref{A1} divides the eigenvalues into the diverging and bounded ones. This assumption is satisfied by the factor model with pervasive factors, which is stated in Assumption \ref{A2}. We say that a factor is pervasive in the sense that it has non-negligible effect on a non-vanishing proportion of individual time-series. Assumptions \ref{A1}-\ref{A2} are crucial for estimating a high-dimensional factor model: they ensure that the space spanned by the principal components in the population level $\bSigma$ is close to the space spanned by the columns of the factor loading matrix $\bB$. Assumption \ref{A3} is a technical condition which is needed to consistently estimate the factors and loadings.
	
	Let $\bSigma=\bGamma\bLambda\bGamma^{'}$, where $\bSigma$ is the covariance matrix of returns that follow factor structure described in equation \eqref{equ1}. Define $\widehat{\bSigma}, \widehat{\bLambda}_q,\widehat{\bGamma}_q$ to be the estimators of $\bSigma,\bLambda,\bGamma$. We further let $\widehat{\bLambda}_q=\text{diag}(\hat{\lambda}_1,\ldots,\hat{\lambda}_q)$ and $\widehat{\bGamma}_q=(\hat{v}_1,\ldots,\hat{v}_q)$ to be constructed by the first $q$ leading empirical eigenvalues and the corresponding eigenvectors of $\widehat{\bSigma}$ and $\widehat{\bB}\widehat{\bB}'=\widehat{\bGamma}_q\widehat{\bLambda}_q\widehat{\bGamma}_{q}^{'}$. Similarly to \cite{fan2018elliptical}, we require the following bounds on the componentwise maximums of the estimators:
	\begin{enumerate}[\textbf{({B}.1)}]
		\item \label{B1} $\norm{\widehat{\bSigma}-\bSigma}_{\text{max}}=\mathcal{O}_P(\sqrt{\log p/T})$,
	\end{enumerate}
	
	\begin{enumerate}[\textbf{({B}.2)}]
		\item \label{B2} $\norm{(\widehat{\bLambda}_q-\bLambda)\bLambda^{-1}}_{\text{max}}=\mathcal{O}_P(\sqrt{\log p/T})$,
	\end{enumerate}
	\begin{enumerate}[\textbf{({B}.3)}]
		\item \label{B3} $\norm{\widehat{\bGamma}_q-\bGamma}_{\text{max}}=\mathcal{O}_P(\sqrt{\log p/(Tp)})$.
	\end{enumerate}
	Assumptions \ref{B1}-\ref{B3} are needed in order to ensure that the first $q$ principal components are approximately the same as the columns of the factor loadings. The estimator $\widehat{\bSigma}$ can be thought of as any ``pilot" estimator that satisfies \ref{B1}. For sub-Gaussian distributions, sample covariance matrix, its eigenvectors and eigenvalues satisfy \ref{B1}-\ref{B3}.
	
	In addition, the following structural assumptions on the model are imposed:
	\begin{enumerate}[\textbf{({C}.1)}]
	\item \label{C1} $\norm{\bSigma}_{\text{max}}=\mathcal{O}(1)$ and $\norm{\bB}_{\text{max}}=\mathcal{O}(1)$.
	\end{enumerate}
\subsection{Convergence of Forecast Combination Weights and MSFE}
To study the properties of the combination weights in \eqref{eq35} and MSFE, we first need to establish the convergence properties of precision matrix produced by Algorithms \ref{alg2}-\ref{alg3}. Let $\omega_{T}\defeq \sqrt{\log p/T} +1/\sqrt{p}$. Also, let $s(\bTheta_{\varepsilon})=\mathcal{O}_P(s_T)$ for some sequence $s_T\in(0,\infty)$ and $d(\bTheta_{\varepsilon})=\mathcal{O}_P(d_T)$ for some sequence $d_T\in(0,\infty)$. The deterministic sequences $s_T$ and $d_T$ will control the sparsity $\bTheta_{\varepsilon}$ for Factor GLASSO. Note that $d_T$ can be smaller than or equal to $s_T$. The reason why we distinguish between these two sequences is to juxtapose it with the sparsity conditions for the Factor MB, where we will only use the analogue of $d_T$ which was defined as $\bar{d}$ at the beginning of this section.

Let $\varrho_{1T}$ be a sequence of positive-valued random variables such that $\varrho_{1T}^{-1}\omega_{T}\xrightarrow{p}0$ and $\varrho_{1T}d_Ts_T\xrightarrow{\text{p}}0$, with $\lambda \asymp \omega_{T}$ (where $\lambda$ is the tuning parameter for the Factor GLASSO in \eqref{e7.6}). \cite{seregina2020FGL} show that under the Assumptions \ref{A1}-\ref{A3}, \ref{B1}-\ref{B3} and \ref{C1}, $\vertiii{\widehat{\bTheta}-\bTheta}_{1}=\mathcal{O}_P(\varrho_{1T}d_Ts_T)$ for Factor GLASSO. Furthermore, let $\varrho_{2T}$ be a sequence of positive-valued random variables such that $\varrho_{2T}^{-1}\omega_{T}\xrightarrow{p}0$ and $\varrho_{2T}\bar{d}^{2}\xrightarrow{\text{p}}0$, with $\lambda_j \asymp \omega_{T}$ (where $\lambda_j$ is the tuning parameter for Factor nodewise regression in \eqref{e21a}). \cite{seregina2020sparse} shows that under the Assumptions \ref{A1}-\ref{A3}, \ref{B1}-\ref{B3}, and \ref{C1}, we have $\vertiii{\widehat{\bTheta}-\bTheta}_1=\mathcal{O}_P(\varrho_{2T}\bar{d}^{2})$. It is interesting to compare the rates for precision matrix obtained by two factor graphical models: if $d_T=s_T$, the rates are similar, whereas if $d_T<s_T$ Factor MB is expected to converge faster. In fact, in high dimensions when $p>T$ and $\omega_{T} \simeq \sqrt{\log p/T}$, Factor MB achieves the minimax rate for this problem (see \cite{cai2016convergence} for the rate expression). 

Having established the convergence rates for precision matrix, we now study the properties of the combination weights and MSFE.
\begin{thm} \label{theor1}
	Assume \ref{A1}-\ref{A3}, \ref{B1}-\ref{B3}, and \ref{C1} hold. 
	\begin{enumerate}
		\item[(i)]If $\varrho_{1T}d_{T}^{2}s_T\xrightarrow{\text{p}}0$, Algorithm \ref{alg2} consistently estimates forecast combination weights in \eqref{eq35}: $\norm{\widehat{\bw}-\bw}_1=\mathcal{O}_P\Big(\varrho_{1T}d_{T}^2s_T\Big)=o_P(1)$.
		\item[(ii)]  If $\varrho_{2T}\bar{d}^{3}\xrightarrow{\text{p}}0$, Algorithm \ref{alg3} consistently estimates forecast combination weights in \eqref{eq35}: $\norm{\widehat{\bw}-\bw}_1=\mathcal{O}_P\Big(\varrho_{2T}\bar{d}^{3}\Big)=o_P(1)$.
	\end{enumerate}
\end{thm}

\begin{thm} \label{theor2}
	Assume \ref{A1}-\ref{A3}, \ref{B1}-\ref{B3}, and \ref{C1} hold. 
	\begin{enumerate}
		\item[(i)]If $\varrho_{1T}d_Ts_T\xrightarrow{\text{p}}0$, Algorithm \ref{alg2} consistently estimates $\text{MSFE}(\bw,\bSigma)$: $	\abs{\frac{\text{MSFE}(\widehat{\bw},\widehat{\bSigma})}{\text{MSFE}(\bw,\bSigma)} -1}=$\\$\mathcal{O}_P(\varrho_{1T}d_Ts_T )=o_P(1)$.
		\item[(ii)]  If $\varrho_{2T}\bar{d}^{2}\xrightarrow{\text{p}}0$, Algorithm \ref{alg3} consistently estimates  $\text{MSFE}(\bw,\bSigma)$: $	\abs{\frac{\text{MSFE}(\widehat{\bw},\widehat{\bSigma})}{\text{MSFE}(\bw,\bSigma)} -1}=\mathcal{O}_P(\varrho_{2T}\bar{d}^{2} )=o_P(1)$.
	\end{enumerate}
\end{thm}
Proofs of Theorems \ref{theor1}-\ref{theor2} can be found in Section 9. Note that the rates of convergence for MSFE and precision matrix $\bTheta$ are the same and both are faster than the combination weight rates in Theorem \ref{theor1}. In contrast to classical graphical models in Algorithms \ref{alg1a}-\ref{alg1b}, the convergence properties of which were examined by \cite{Sara2018} among others, the rates in Theorems \ref{theor1}-\ref{theor2} depend on the sparsity of $\bTheta_{\varepsilon}$ rather than of $\bTheta$. This means that instead of assuming that many partial correlations of forecast errors $\be_{t}$ are negligible, which is not realistic under the factor structure, we impose a milder restriction requiring many partial correlations of $\bvarepsilon_{t}$ to be negligible once the common components have been taken into account. Similarly to the comparison of precision matrix $\bTheta$ obtained by two graphical models, if $d_T<s_T$ Factor MB is expected to converge faster for combination weights and MSFE. In our simulations the rates of Factor Graphical models are comparable, whereas an empirical application shows that for most macroeconomic series that we studied Factor GLASSO outperforms Factor MB. This suggests that for macroeconomic forecasting using weighted penalized log-likelihood and running $p$ coupled LASSO problems for estimating precision matrix is preferable to fitting $p$ separate LASSO regressions using each variable as the response and the others as predictors.
	\section{Monte Carlo}
	We divide the simulation results into two subsections. In the first subsection we study the consistency of the Factor GLASSO  and Factor MB for estimating precision matrix and the combination weights. In the second subsection we evaluate the out-of-sample forecasting performance of combined forecasts based on the Factor Graphical models from Algorithms \ref{alg2}-\ref{alg3} in terms of the mean-squared forecast error. We compare the performance of forecast combinations based on the factor models with equal-weighted (EW) forecast combination, forecast combinations using GLASSO and nodewise regression from Algorithms \ref{alg1a}-\ref{alg1b}. Similarly to the literature on graphical models, all exercises use 100 Monte Carlo simulations.
	\subsection{Consistent Estimation of forecast combination weights based on FGM}
	We consider sparse Gaussian graphical models which may be fully specified by a precision matrix $\bTheta_0$. Therefore, the random sample is distributed as $\be_t=(e_{1t},\ldots,e_{pt})' \sim \mathcal{N}(0,\bSigma_{0})$, where $\bTheta_0=(\bSigma_{0})^{-1}$  for $t=1,\ldots, T, \ j=1,\ldots, p$. Let $\widehat{\bTheta}$ be the precision matrix estimator. We show consistency of the Factor GLASSO (Algorithm \ref{alg2}) and Factor MB (Algorithm \ref{alg3}), in (i) the operator norm, $\vertiii{\widehat{\bTheta}-\bTheta_{0}}_{2}$, (ii) $\ell_{1}$/$\ell_{1}$-matrix norm, $\vertiii{\widehat{\bTheta}-\bTheta_{0}}_{1}$, and (iii) in $\ell_1$-vector norm for the combination weights, $\norm{\widehat{\bw}-\bw}_1$, where $\bw$ is given by \eqref{eq13}.
	
	The forecast errors are assumed to have the following structure:
	\begin{align} \label{e51}
	&\underbrace{\be_t}_{p \times 1}=\bB \underbrace{\bf_t}_{q\times 1}+ \ \bvarepsilon_t,\quad t=1,\ldots,T\\
	&\bf_{t}=\phi_f\bf_{t-1}+\bzeta_t,
	\end{align}
	where $\be_{t}$ is a $p \times 1$ vector of forecast errors following $\mathcal{N}(\bm{0},\bSigma)$, $\bf_{t}$ is a $q \times 1$ vector of factors, $\bB$ is a $p\times q$ matrix of factor loadings, $\phi_f$ is an autoregressive parameter in the factors which is a scalar for simplicity, $\bzeta_t$ is a $q \times 1$ random vector with each component independently following  $\mathcal{N}(0,\sigma^{2}_{\zeta})$, $\bvarepsilon_t$ is a $p \times 1$ random vector following $\mathcal{N}(0,\bSigma_{\varepsilon})$, with sparse $\bTheta_{\varepsilon}$ that has a random graph structure described below. To create $\bB$ in \eqref{e51} we take the first $q$ columns of an upper triangular matrix from a Cholesky decomposition of the $p \times p$ Toeplitz matrix parameterized by $\rho$: that is, $\bB = (b)_{ij}$, where $(b)_{ij}=\rho^{\abs{i-j}}$, $i,j\in \{1,\ldots,p\}$.
	 We set $\rho = 0.2$, $\phi_f = 0.2$ and $\sigma^{2}_{\zeta} = 1$. The specification in \eqref{e51} leads to the low-rank plus sparse decomposition of the covariance matrix:
	\begin{align}
	\E{\be_t\be'_t}=\bSigma=\bB\bSigma_{f}\bB'+ \bSigma_{\varepsilon}.
	\end{align}
	
\noindent	When $\bSigma_{\varepsilon}$ has a sparse inverse $\bTheta_{\varepsilon}$, it leads to the low-rank plus sparse decomposition of the precision matrix $\bTheta$, such that $\bTheta$ can be expressed as a function of the low-rank $\bTheta_{f}$ plus sparse $\bTheta_{\varepsilon}$.
	
	We consider the following setup: let $p = T^{\delta}$, $\delta = 0.85$, $q = 2(\log(T))^{0.5}$ and $T = \lbrack 2^{\kappa} \rbrack, \ \text{for} \ \kappa=7,7.5,8,\ldots,9.5$. Our setup allows the number of individual forecasts, $p$, and the number of common factors in the forecast errors, $q$, to increase with the sample size, $T$.
	
	 A sparse precision matrix of the idiosyncratic components $\bTheta_{\varepsilon}$ is constructed as follows: we first generate the adjacency matrix using a random graph structure. Define a $p \times p$ adjacency matrix $\bA_{\varepsilon}$ which represents the structure of the graph:
	\begin{align}
	a_{\varepsilon,ij}=\begin{cases}
	1, & \text{for} \ i\neq j\ \ \text{with probability $\pi$},\\
	0, & \text{otherwise,}
	\end{cases}
	\end{align}
	where $a_{\varepsilon,ij}$ denotes the $i,j$-th element of the adjacency matrix $\bA_{\varepsilon}$. We set $a_{\varepsilon,ij} = a_{\varepsilon,ji}=1, \ \text{for} \ i\neq j$ with probability $\pi$, and $0$ otherwise. Such structure results in $s_T = p(p-1)\pi/2$ edges in the graph. To control sparsity, we set $\pi = 1/(pT^{0.8})$, which makes $s_T = \mathcal{O}(T^{0.05})$. The adjacency matrix has all diagonal elements equal to zero. Hence, to obtain a positive definite precision matrix we apply the procedure described in \cite{HUGE}: using their notation, $\bTheta_{\varepsilon}=\bA_{\varepsilon}\cdot v+\bI(\abs{\tau}+0.1+u)$, where $u>0$ is a positive number added to the diagonal of the precision matrix to control the magnitude of partial correlations, $v$ controls the magnitude of partial correlations with $u$, and $\tau$ is the smallest eigenvalue of $\bA_{\varepsilon}\cdot v$. In our simulations we use $u=0.1$ and $v=0.3$.
	
	Figures \ref{f18}-\ref{f19} show the averaged (over Monte Carlo simulations) errors of the estimators of the precision matrix $\bTheta$ and the optimal combination weight versus the sample size $T$ in the logarithmic scale (base 2). The estimate of the precision matrix of the EW forecast combination is obtained using the fact that diagonal covariance and precision matrices imply equal weights. To determine the values of the diagonal elements we use the shrinkage intensity coefficient calculated as the average of the eigenvalues of the sample covariance matrix of the forecast errors (see \cite{Ledoit2004}). As evidenced by Figures \ref{f18}-\ref{f19}, Factor GLASSO and Factor MB demonstrate superior performance over EW and non-factor based models (GLASSO and MB). Furthermore, our method achieves lower estimation error in the combination weights \eqref{e4.9}, which leads to lower risk of the combined forecast as shown in \eqref{e4.8}. Interestingly, even though the precision matrix estimated using Factor MB has faster convergence rate in $\vertiii{\cdot}_2$ and $\vertiii{\cdot}_1$ norms as compared to Factor GLASSO, the weights estimated using Factor GLASSO converge faster. Also, note that the precision matrix estimated using the EW method also shows good convergence properties. However, in terms of estimating the combination weight, the performance of EW does not exhibit convergence properties. This is in agreement with previously reported findings (\cite{smith2009simple}) that equal weights are not theoretically optimal, however, as demonstrated in the next subsection, the EW combination still leads to a relatively good performance in terms of MSFE although the FGM-based combinations outperform it. 
	\subsection{Comparing Performance of forecast combinations based on FGM}	
We consider the standard forecasting model in the literature (e.g., \cite{Stock2002}), which uses the factor structure of the high dimensional predictors.	
Suppose the data is generated from the following data generating process (DGP):
\begin{align}
&\bx_t=\bLambda\bg_t+\bv_t,\label{e39}\\
&\bg_{t}=\phi\bg_{t-1}+\bxi_t,\label{e40}\\
&y_{t+1}=\bg'_t\balpha+\sum_{s=1}^{\infty}\theta_{s}\epsilon_{t+1-s}+\epsilon_{t+1},\label{e48}
\end{align}
where $y_{t+1}$ is a univariate series of our interest in forecasting, $\bx_t$ is an $N \times 1$ vector of regressors (predictors), $\bbeta$ is an $N \times 1$ parameter vector, $\bg_{t}$ is an $r \times 1$ vector of factors, $\bLambda$ is an $N\times r$ matrix of factor loadings, $\bv_t$ is an $N \times 1$ random vector following $\mathcal{N}(0,\sigma^{2}_{v})$, $\phi$ is an autoregressive parameter in the factors which is a scalar for simplicity, $\bxi_t$ is an $r \times 1$ random vector with each component independently following $\mathcal{N}(0,\sigma^{2}_{\xi})$, $\epsilon_{t+1}$ is a random error following $\mathcal{N}(0,\sigma^{2}_{\epsilon})$, and $\balpha$ is an $r\times 1$ parameter vector which is drawn randomly from $\mathcal{N}(1,1)$. We set $\sigma_{\epsilon}=1$. The coefficients $\theta_{s}$ are set according to the rule 
\begin{align} \label{e41}
\theta_{s}=(1+s)^{c_1}c_{2}^{s}, 
\end{align}
as in \cite{HANSEN2008}. We set $c_1\in\{0,0.75\}$ and $c_2 \in \{0.6, 0.7, 0.8, 0.9\}$. We generate $r$ factors using \eqref{e40} with a grid of 10 different AR(1) coefficients $\phi$ equidistant between $0$ and $0.9$. To create $\bLambda$ in \eqref{e39} we take the first $r$ rows of an upper triangular matrix from a Cholesky decomposition of the $N \times N$ Toeplitz matrix parameterized by $\rho$. We consider a grid of 10 different values of $\rho$ equidistant between $0$ and $0.9$.

One-step ahead forecasts are estimated from the factor-augmented autoregressive (FAR) models of orders $k,l$, denoted as  FAR($k,l$):
\begin{align}\label{e42}
&\hat{y}_{t+1}=\hat{\mu}
+\hat{\kappa}_1\hat{g}_{1,t}+\cdots+ \hat{\kappa}_k\hat{g}_{k,t}
+\hat{\psi}_1y_{t}+\cdots+\hat{\psi}_ly_{t+1-l},
\end{align}
where the factors $({\hat{g}_{1,t},\ldots,\hat{g}_{k,t}})$ are estimated from equation (\ref{e39}).
We consider the FAR models of various orders, with $k=1,\ldots,K$ and $l=1,\ldots,L$.
We also consider the models without any lagged $y$ or any factors. Therefore, the total number of forecasting models is $p \equiv (1+K)\times(1+L)$, which includes the forecasting models using naive average or no factors.

The total number of observations is $T$, and the number of observations in the regression period (the train sample) is set to be the first half of the sample, $t=1,\ldots,m \equiv T/2$, to leave the second half of the sample, $t=m+1,\ldots,T$, for the out-of-sample evaluation (the test sample). We roll the estimation window over the test sample of the size $n\equiv T-m$, to update all the estimates in each point of time $t=1,\ldots,m$. Recall that $q$ denotes the number of factors in the forecast errors as in equation (\ref{equ1}). We first examine the properties of the combined forecasts based on the Factor Graphical models when $T$ and $p$ vary and compare their performance with the combined forecasts based on the GLASSO, MB and EW forecasts.

We consider a low-dimensional setup to demonstrate the advantage of using FGM even when the number of forecasts, $p$, is small relative to the sample size, $T$: (1) in such scenario EW has an advantage since there are not many models to combine and assigning equal weights should produce satisfactory performance, and (2) non-factor based models have the advantage over the models that estimate factors due to the estimation errors. As a result, this framework with the low-dimensional setup is favorable to EW and non-factor based models. Figure \ref{fig2} shows the MSFE for different sample sizes and fixed parameters: we report the results for two values of $c_1\in\{0,0.75\}$. As evidenced from Figure \ref{fig2}, the models that use the factor structure outperform EW combination and non-factor based counterparts for both values of $c_1$. We see that Factor GLASSO, in general, has lower MSFE than Factor MB. This finding is further supported by our empirical application in Section 7.

In Appendix \ref{appendixA} we examine the sensitivity of the competing models with respect to variation in the DGP parameters such as number of predictors $N$, values of $c_2$, $\phi$, the strength of factor loadings $\rho$, and the number of factors $q$. We conclude that Factor Graphical Models outperform equally-weighted combinations and the graphical models without factors.

	\section{Application of FGM for Macroeconomic Forecasting}
	 An empirical application to forecasting macroeconomic time series in big data environment highlights the advantage of both Factor Graphical models described in Algorithms \ref{alg2}-\ref{alg3} in comparison with the existing methods of forecast combination. We use a large monthly frequency macroeconomic database of \cite{DataMcCracken}, who provide a comprehensive description of the dataset and 128 macroeconomic series. We consider the time period 1960:01-2020:07 with the total number of observations $T=726$, the training sample consists of $m=120$ observations, and the test sample $n\equiv T-m-h+1$, where $h$ is the forecast horizon. We roll the estimation window over the test sample to update all the estimates in each point of time $t=m, \ldots,T-h$. We estimate $h$-step ahead forecasts from FAR($k,l$) which were defined in \eqref{e42} with $k=0, 1, \ldots, K=9$, and $l=0, 1,\ldots, L=11$. The total number of forecasting models is $p=120$. The optimal number of factors in the forecast errors (denoted as $q$ in equation \eqref{equ1}) is chosen using the standard data-driven method that uses the information criterion IC1 described in \cite{Bai2002}. We note that in the majority of the cases the optimal number of factors was estimated to be equal to 1.
	 
	 Table \ref{tab1} compares the performance of the Factor GLASSO and Factor MB with the competitors for predicting seven representative macroeconomic indicators of the US economy: monthly industrial production (INDPRO), S\&P500 composite index (S\&P500), Consumer Price Index (CPIAUCSL), real personal consumption (DPCERA3MO86SBEA), M1 money stock (M1SL), civilian unemployment rate (UNRATE), and the effective federal funds rate (FEDFUNDS) using 127 remaining macroeconomic series. Let $\{Y_t\}_{t=1}^{T}$ be the series of interest for forecasting. Similarly to \cite{stevanovic2020machine}, for INDPROD, S\&P500, CPI, Real Personal Consumption and M1 Money Stock we forecast the average growth rate (with logs):
	\begin{equation}
	y_{t+h} = \frac{1}{h} \ln (Y_{t+h}/Y_t).
	\end{equation}
	  For UNRATE we forecast the average change (without logs):
	\begin{equation}
	y_{t+h} = \frac{1}{h} (Y_{t+h}/Y_t).
	\end{equation}	  
	  And for FEDFUNDS we forecast the log of the series:
	  	\begin{equation}
	  y_{t+h} = \ln (Y_{t+h}).
	  \end{equation}
	  	  
	   Table \ref{tab1} reports MSFEs of the competing methods with the smallest MSFE in each row in bold font. As evidenced from Table \ref{tab1}, our methods outperform EW, GLASSO and nodewise regression: accounting for the factor structure results in lower MSFE. Therefore, the FGM framework developed in this paper leads to the superior performance of the combined forecast as compared to EW model even when the models/experts do not contain a lot of unique information. Our empirical application demonstrates that this finding does not originate from the difference in the performance of EW vs graphical models: as evidenced from Table \ref{tab1}, the performance of GLASSO is worse than that of EW for the FEDFUNDS series, whereas Factor GLASSO outperforms EW. A similar pattern is observed in the performance of nodewise regression for M1 Money Stock. Therefore, the improvement in the combined forecast comes from incorporating the factor structure of the forecast errors into the graphical models. Note that in contrast with EW and non-factor based methods, the performance of Factor GLASSO and Factor MB does not deteriorate significantly when the forecast horizon, $h$, increases. Notice, however, that Factor Graphical Models tend to perform better for $h\geq 2$. In other words, accounting for common factors in forecast errors has greater benefit for longer horizons. Finally, for most series Factor GLASSO outperforms Factor MB, suggesting that for macroeconomic forecasting using weighted penalized log-likelihood and running $p$ coupled LASSO problems for estimating precision matrix is preferable to fitting $p$ separate LASSO regressions using each variable as the response and the others as predictors.

	\section{Conclusions}
	In this paper we overcome the challenge of using graphical models under the factor structure and  provide a simple framework that allows practitioners to combine a large number of forecasts when experts tend to make common mistakes. Our new approach to forecast combinations breaks down forecast errors into common and unique parts which improves the accuracy of the combined forecast. The proposed algorithms, Factor Graphical Models, are shown to consistently estimate forecast combination weights and MSFE. Extensive simulations and empirical applications to macroeconomic forecasting in big data environment reveal that FGM outperforms equal-weighted forecasts and combined forecasts produced using graphical models without factors. With the superior performance observed at all forecast horizons, we find that the greater benefit from accounting for the common factors is evidenced at longer horizons.

	\section{Appendix} \label{appendixAtheor}
		\begin{spacing}{1.65}
	In this section we collected the proofs of Theorems \ref{theor1}-\ref{theor2}. We first present a Lemma which is used in the theoretical derivations.
	\begin{lem}\label{lemma1}  Let $l\in\{1,2\}\defeq\{\text{Factor GLASSO}, \text{Factor MB}\}$.
		~
		\begin{enumerate}[label=(\alph*)]
			\item $\vertiii{\bTheta}_1=\mathcal{O}(\kappa_{1l})$, where $\kappa_{1l}=d_T$ if $l=1$ which corresponds to Factor GLASSO, and $\kappa_{1l}=\bar{d}$ if $l=2$ which corresponds to Factor MB. This will be further abbreviated as $\kappa_{1l} \in \{d_T, \bar{d}\}_{l=1,2}$.
			\item $a \geq C_0>0$, where $a$ was defined in Section 3 and $C_0$ is a positive constant representing the minimal eigenvalue of $\bTheta$.
			\item $\abs{\widehat{a}-a}=\mathcal{O}_P(\kappa_{2l})$, where $\widehat{a}$ was defined in Section 3 and  $\kappa_{2l} \in \{\varrho_{1T}d_Ts_T, \varrho_{2T}\bar{d}^2 \}_{l=1,2}$.	
		\end{enumerate}
	\end{lem}
	\begin{proof}
		~
		\begin{enumerate}[label=(\alph*)]
			\item To prove part (a) we use the following matrix inequality which holds for any $\bA \in \mathcal{S}_p$:
			\begin{equation}\label{a3}
			\vertiii{\bA}_1=\vertiii{\bA}_{\infty}\leq \sqrt{d(\bA)}\vertiii{\bA}_2,
			\end{equation}
			where $d(\bA)$ was defined at the beginning of Section 5. The proof of \eqref{a3} is a straightforward consequence of the Schwarz inequality.
			
			Sherman-Morrison-Woodbury formula together with \eqref{a3} and Assumptions \ref{B1}-\ref{B3} yield:
			\begin{align}\label{a4}
			\vertiii{\bTheta}_1 &\leq \vertiii{\bTheta_{\varepsilon}}_1+\vertiii{\bTheta_{\varepsilon}\bB\lbrack\bTheta_{f}+\bB'\bTheta_{\varepsilon}\bB \rbrack^{-1}\bB'\bTheta_{\varepsilon}}_1 \nonumber\\
			&=\mathcal{O}(\sqrt{\kappa_{1l}})+\mathcal{O}\Big(\sqrt{\kappa_{1l}}\cdot p \cdot \frac{1}{p} \cdot \sqrt{\kappa_{1l}}\Big) = \mathcal{O}(\kappa_{1l}).
			\end{align}
			\item Assumption \ref{A1} states that the minimal eigenvalue of $\bTheta$ is bounded away from zero, hence,
			\begin{equation*}
			a=\biota'_p\bTheta\biota_p/p\geq  C_0>0.
			\end{equation*}
			\item Using the H\"{o}lders inequality, we have
			\begin{align*} 
			\abs{\widehat{a}-a}=\abs{\frac{\biota'_{p}(\widehat{\bTheta}-\bTheta)\biota_p}{p}}\leq \frac{\norm{(\widehat{\bTheta}-\bTheta)\biota_p }_1\norm{\biota_p}_{\infty}}{p} &\leq \vertiii{ \widehat{\bTheta}-\bTheta}_1 \\
			&=\mathcal{O}_P(\kappa_{2l})=o_P(1),
			\end{align*}
			where the last rate is obtained using the assumptions of Theorem \ref{theor1}.
		\end{enumerate}		
	\end{proof}	
	\subsection{Proof of Theorem \ref{theor1}}
	First, note that the forecast combination weight can be written as
	\begin{align*}
	\widehat{\bw}-\bw &= \frac{\Big((a\widehat{\bTheta}\biota_p) - (\hat{a}\bTheta\biota_p)  \Big)/p}{\hat{a}a}\\
	&=\frac{ \Big((a\widehat{\bTheta}\biota_p) - (a\bTheta\biota_p) + (a\bTheta\biota_p) - (\hat{a}\bTheta\biota_p) \Big)/p  }{\hat{a}a}.
	\end{align*}
	As shown in \cite{Caner2019}, the above can be rewritten as
	\begin{align}\label{a1}
	\norm{\widehat{\bw}-\bw}_1\leq \frac{a\frac{\norm{(\widehat{\bTheta}-\bTheta)\biota_p}_1 }{p}+\abs{a-\widehat{a} }\frac{\norm{\bTheta\biota_p}_1}{p}}{\abs{\widehat{a}}a}.
	\end{align}
	Prior to bounding the terms in \eqref{a1}, we first present an inequality which is used in the derivations. Let $\bA \in \mathbb{R}^{p\times p}$ and $\bv \in \mathbb{R}^{p\times1}$. Also, let $\bA_{j}$ and $\bA'_{j}$ be a $p \times 1$ and $1\times p$ row and column vectors in $\bA$, respectively.
	\begin{align}\label{a2}
	\norm{\bA\bv}_1 &= \abs{\bA'_{1}\bv}+\ldots+\abs{\bA'_{p}\bv} \leq \norm{\bA_1}_1\norm{\bv}_{\infty}+\ldots+\norm{\bA_p}_1\norm{\bv}_{\infty}\\
	&=\Bigg(\sum_{j=1}^{p} \norm{\bA_j}_1 \Bigg)\norm{\bv}_{\infty} \leq p \max_{j}\abs{\bA_{j}}_1\norm{\bv}_{\infty}.\nonumber
	\end{align}
	H\"{o}lders inequality was used to obtain each inequality in \eqref{a2}. If $\bA \in \mathcal{S}_{p}$, then the last expression can be further reduced to $p \vertiii{\bA}_1\norm{\bv}_{\infty}$.
	
	Let us now bound the right-hand side of \eqref{a1}. In the numerator we have:
	\begin{align}\label{a5}
	\frac{\norm{(\widehat{\bTheta}-\bTheta)\biota_p}_1 }{p} \leq \vertiii{\bTheta}_1=\mathcal{O}_P(\kappa_{3l}),
	\end{align}
	where $\kappa_{3l} \in \{\varrho_{1T}d_Ts_T, \varrho_{2T}\bar{d}^{2}\}_{l=1,2}$, the rates were derived in \cite{seregina2020FGL,seregina2020sparse} as discussed at the beginning of Section 5, and the inequality follows from \eqref{a2}.
	\begin{equation}\label{a6}
	\frac{\norm{\bTheta\biota_p}_1}{p}\leq \vertiii{\bTheta}_1=\mathcal{O}(\kappa_{1l}),
	\end{equation}
	where the rate follows from Lemma \ref{lemma1} (a) and the inequality is obtained from \eqref{a2}. Combining \eqref{a5}, \eqref{a6}, and Lemma \ref{lemma1} (c) we get:
	\begin{align}
	a\frac{\norm{(\widehat{\bTheta}-\bTheta)\biota_p}_1 }{p}+\abs{a-\widehat{a} }\frac{\norm{\bTheta\biota_p}_1}{p} = \mathcal{O}(1)\cdot\mathcal{O}_P(\kappa_{3l})+\mathcal{O}_P(\kappa_{2l})\cdot \mathcal{O}(\kappa_{1l}) = \mathcal{O}_P(\kappa_{4l})=o_P(1),
	\end{align}
	where $\kappa_{4l} \in \{\varrho_{1T}d_{T}^{2}s_T, \varrho_{2T}\bar{d}^{3}\}_{l=1,2}$ and the last equality holds under the assumptions of Theorem \ref{theor1}.
	
	For the denominator of \eqref{a1} it easy to see that $\abs{\widehat{a}}a=\mathcal{O}_P(1)$ using the results of Lemma \ref{lemma1} (b).
	\subsection{Proof of Theorem \ref{theor2}}
	Using Lemma \ref{lemma1} (b)-(c), we get
	\begin{align*}
	\abs{\frac{\hat{a}^{-1}}{a^{-1}}-1}=\frac{\abs{a-\hat{a}}}{\abs{\hat{a}}} = \mathcal{O}_P(\kappa_{2l})=o_P(1),
	\end{align*}
	where the last rate is obtained using the assumptions of Theorem \ref{theor2}.
\end{spacing}
	\cleardoublepage	
	\phantomsection	
	\addcontentsline{toc}{section}{References}	
	\setlength{\baselineskip}{14pt}
	\bibliographystyle{apalike}
	\bibliography{LeeSereginaForecasting}

\begin{thebibliography}{}

\bibitem[Bai, 2003]{Bai2003}
Bai, J. (2003).
\newblock Inferential theory for factor models of large dimensions.
\newblock {\em Econometrica}, 71(1):135--171.

\bibitem[Bai and Ng, 2002]{Bai2002}
Bai, J. and Ng, S. (2002).
\newblock Determining the number of factors in approximate factor models.
\newblock {\em Econometrica}, 70(1):191--221.

\bibitem[Barigozzi et~al., 2018]{Brownlees2018EJS}
Barigozzi, M., Brownlees, C., and Lugosi, G. (2018).
\newblock Power-law partial correlation network models.
\newblock {\em Electronic Journal of Statistics}, 12(2):2905--2929.

\bibitem[Bates and Granger, 1969]{GrangerBatesWeights}
Bates, J.~M. and Granger, C. W.~J. (1969).
\newblock The combination of forecasts.
\newblock {\em Operations Research}, 20(4):451--468.

\bibitem[Brownlees et~al., 2018]{Brownlees2018JAE}
Brownlees, C., Nualart, E., and Sun, Y. (2018).
\newblock Realized networks.
\newblock {\em Journal of Applied Econometrics}, 33(7):986--1006.

\bibitem[Cai et~al., 2016]{cai2016convergence}
Cai, T.~T., Liu, W., Zhou, H.~H., et~al. (2016).
\newblock Estimating sparse precision matrix: Optimal rates of convergence and
  adaptive estimation.
\newblock {\em Annals of Statistics}, 44(2):455--488.

\bibitem[Callot et~al., 2019]{Caner2019}
Callot, L., Caner, M., \"Onder, A.~O., and Ulaşan, E. (2019).
\newblock A nodewise regression approach to estimating large portfolios.
\newblock {\em Journal of Business \& Economic Statistics}, 0(0):1--12.

\bibitem[Callot et~al., 2017]{Callot2017}
Callot, L. A.~F., Kock, A.~B., and Medeiros, M.~C. (2017).
\newblock Modeling and forecasting large realized covariance matrices and
  portfolio choice.
\newblock {\em Journal of Applied Econometrics}, 32(1):140--158.

\bibitem[Chamberlain and Rothschild, 1983]{Chamberlain}
Chamberlain, G. and Rothschild, M. (1983).
\newblock Arbitrage, factor structure, and mean-variance analysis on large
  asset markets.
\newblock {\em Econometrica}, 51(5):1281--1304.

\bibitem[Chan et~al., 1999]{StockChan1999}
Chan, Y.~L., Stock, J.~H., and Watson, M.~W. (1999).
\newblock A dynamic factor model framework for forecast combination.
\newblock {\em Spanish Economic Review}, 1(2):91--121.

\bibitem[Chen and Chen, 2008]{Chen_EBIC}
Chen, J. and Chen, Z. (2008).
\newblock Extended bayesian information criteria for model selection with large
  model spaces.
\newblock {\em Biometrika}, 95(3):759--771.

\bibitem[Clemen, 1989]{CLEMEN1989}
Clemen, R.~T. (1989).
\newblock Combining forecasts: A review and annotated bibliography.
\newblock {\em International Journal of Forecasting}, 5(4):559--583.

\bibitem[Connor and Korajczyk, 1988]{Connor1988}
Connor, G. and Korajczyk, R.~A. (1988).
\newblock Risk and return in an equilibrium {APT}: application of a new test
  methodology.
\newblock {\em Journal of Financial Economics}, 21(2):255--289.

\bibitem[Coulombe et~al., 2020]{stevanovic2020machine}
Coulombe, P.~G., Leroux, M., Stevanovic, D., and Surprenant, S. (2020).
\newblock How is machine learning useful for macroeconomic forecasting?
\newblock {\em arXiv:2008.12477}.

\bibitem[Diebold and Shin, 2019]{DIEBOLD2018}
Diebold, F. and Shin, M. (2019).
\newblock Machine learning for regularized survey forecast combination:
  Partially-egalitarian lasso and its derivatives.
\newblock {\em International Journal of Forecasting}, 35(4):1679--1691.

\bibitem[Fan et~al., 2018]{fan2018elliptical}
Fan, J., Liu, H., and Wang, W. (2018).
\newblock Large covariance estimation through elliptical factor models.
\newblock {\em The Annals of Statistics}, 46(4):1383--1414.

\bibitem[Foygel and Drton, 2010]{EBIC}
Foygel, R. and Drton, M. (2010).
\newblock Extended bayesian information criteria for gaussian graphical models.
\newblock In {\em Proceedings of the 23rd International Conference on Neural
  Information Processing Systems - Volume 1}, NIPS, pages 604--612, USA. Curran
  Associates Inc.

\bibitem[Friedman et~al., 2008]{GLASSO}
Friedman, J., Hastie, T., and Tibshirani, R. (2008).
\newblock {Sparse inverse covariance estimation with the Graphical Lasso}.
\newblock {\em Biostatistics}, 9(3):432--441.

\bibitem[Hansen, 2008]{HANSEN2008}
Hansen, B.~E. (2008).
\newblock Least-squares forecast averaging.
\newblock {\em Journal of Econometrics}, 146(2):342--350.

\bibitem[Hautsch et~al., 2012]{Hautsch2012}
Hautsch, N., Kyj, L.~M., and Oomen, R. (2012).
\newblock A blocking and regularization approach to high-dimensional realized
  covariance estimation.
\newblock {\em Journal of Applied Econometrics}, 27(4):625--645.

\bibitem[Jankov\'{a} and van~de Geer, 2018]{Sara2018}
Jankov\'{a}, J. and van~de Geer, S. (2018).
\newblock {I}nference in high-dimensional graphical models.
\newblock \textit{Handbook of Graphical Models}, Chapter 14, pages 325--351.
  CRC Press.

\bibitem[Koike, 2020]{koike2019biased}
Koike, Y. (2020).
\newblock De-biased graphical lasso for high-frequency data.
\newblock {\em Entropy}, 22(4):456.

\bibitem[Ledoit and Wolf, 2004]{Ledoit2004}
Ledoit, O. and Wolf, M. (2004).
\newblock A well-conditioned estimator for large-dimensional covariance
  matrices.
\newblock {\em Journal of Multivariate Analysis}, 88(2):365--411.

\bibitem[Lee and Seregina, 2020]{seregina2020FGL}
Lee, T.-H. and Seregina, E. (2020).
\newblock Optimal portfolio using factor graphical lasso.
\newblock {\em arXiv:2011.00435}.

\bibitem[Liu et~al., 2010]{STARS}
Liu, H., Roeder, K., and Wasserman, L. (2010).
\newblock Stability approach to regularization selection (stars) for high
  dimensional graphical models.
\newblock In {\em Proceedings of the 23rd International Conference on Neural
  Information Processing Systems - Volume 2}, NIPS'10, pages 1432--1440, USA.
  Curran Associates Inc.

\bibitem[McCracken and Ng, 2016]{DataMcCracken}
McCracken, M.~W. and Ng, S. (2016).
\newblock {FRED-MD}: A monthly database for macroeconomic research.
\newblock {\em Journal of Business \& Economic Statistics}, 34(4):574--589.

\bibitem[Meinshausen and Bühlmann, 2006]{meinshausen2006}
Meinshausen, N. and Bühlmann, P. (2006).
\newblock High-dimensional graphs and variable selection with the {L}asso.
\newblock {\em The Annals of Statistics}, 34(3):1436--1462.

\bibitem[Meinshausen and B{\"u}hlmann, 2010]{MeinshausenCV}
Meinshausen, N. and B{\"u}hlmann, P. (2010).
\newblock Stability selection.
\newblock {\em Journal of the Royal Statistical Society, Series B},
  72:417--473.

\bibitem[Seregina, 2020]{seregina2020sparse}
Seregina, E. (2020).
\newblock A basket half full: Sparse portfolios.
\newblock {\em arXiv preprint arXiv:2011.04278}.

\bibitem[Smith and Wallis, 2009]{smith2009simple}
Smith, J. and Wallis, K.~F. (2009).
\newblock A simple explanation of the forecast combination puzzle.
\newblock {\em Oxford Bulletin of Economics and Statistics}, 71(3):331--355.

\bibitem[Stock and Watson, 2002]{Stock2002}
Stock, J.~H. and Watson, M.~W. (2002).
\newblock Forecasting using principal components from a large number of
  predictors.
\newblock {\em Journal of the American Statistical Association},
  97(460):1167--1179.

\bibitem[van~de Geer et~al., 2014]{Buhlmann2014}
van~de Geer, S., Buhlmann, P., Ritov, Y., and Dezeure, R. (2014).
\newblock On asymptotically optimal confidence regions and tests for
  high-dimensional models.
\newblock {\em The Annals of Statistics}, 42(3):1166--1202.

\bibitem[Zhao et~al., 2012]{HUGE}
Zhao, T., Liu, H., Roeder, K., Lafferty, J., and Wasserman, L. (2012).
\newblock The {HUGE} package for high-dimensional undirected graph estimation
  in {R}.
\newblock {\em Journal of Machine Learning Research}, 13(1):1059--1062.

\bibitem[Zhu and Cribben, 2018]{Zhu_STARS}
Zhu, Y. and Cribben, I. (2018).
\newblock Sparse graphical models for functional connectivity networks: Best
  methods and the autocorrelation issue.
\newblock {\em Brain Connectivity}, 8(3):139--165.
\newblock PMID: 29634321.

\end{thebibliography}
\cleardoublepage
\phantomsection
\addcontentsline{toc}{section}{Figures}
	\begin{figure}[!htb]
	\centering
	\includegraphics[width=.95\linewidth]{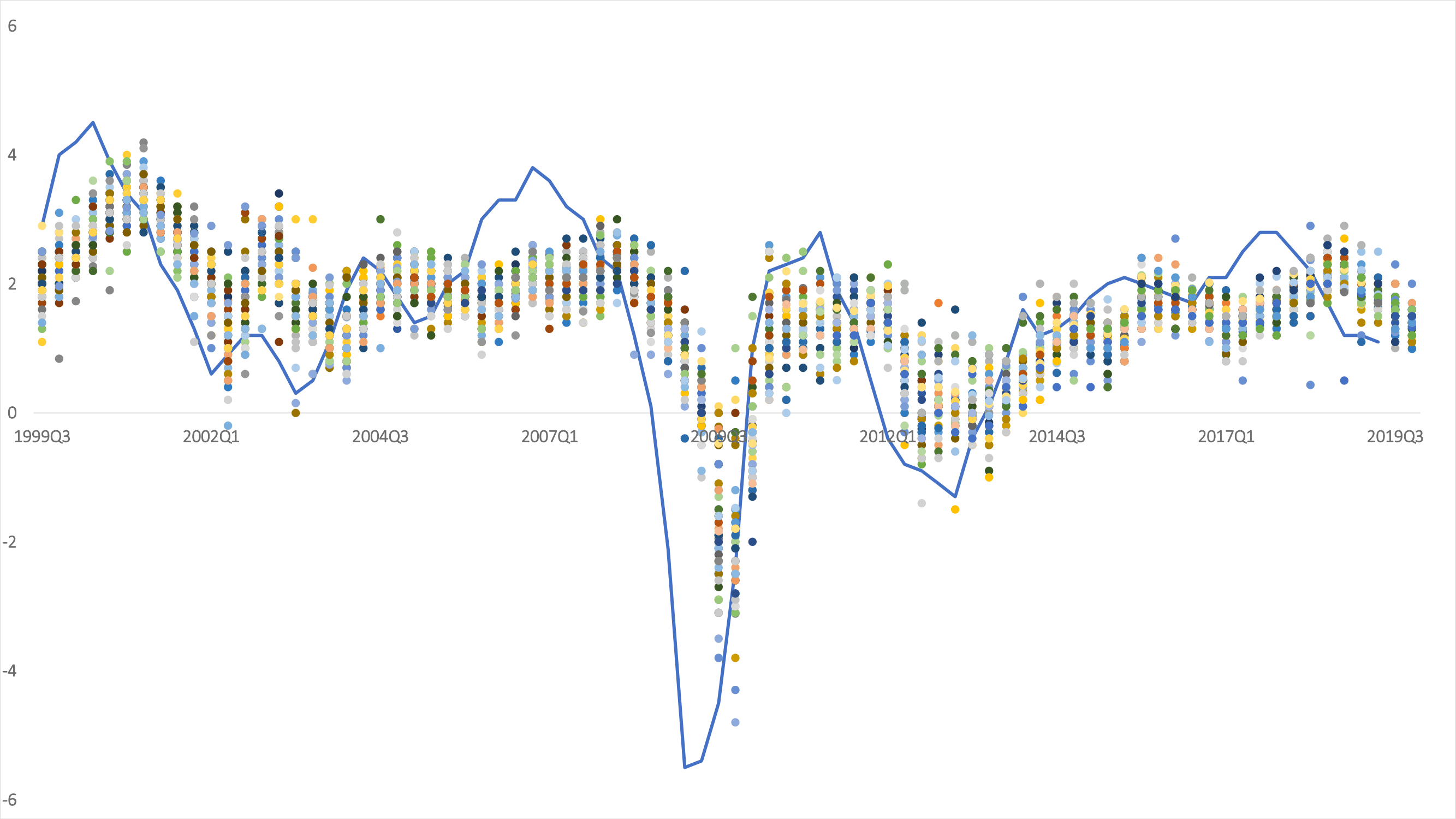}
	\bigskip
	\caption{\textbf{The European Central Bank's (ECB) Survey of Professional Forecasters (SPF)}. Each circle denotes the forecast of each professional forecaster in the SPF for the quarterly 1-year-ahead forecasts of Euro-area real GDP growth, year-on-year percentage change. Actual series is the blue line. \textit{Source: \href{https://www.ecb.europa.eu/stats/ecb_surveys/survey_of_professional_forecasters/html/index.en.html}{European Central Bank}.}}
	\bigskip
	\label{fig1}
\end{figure}
	\begin{figure}[!htb]
	\centering
	\includegraphics[width=.99\linewidth]{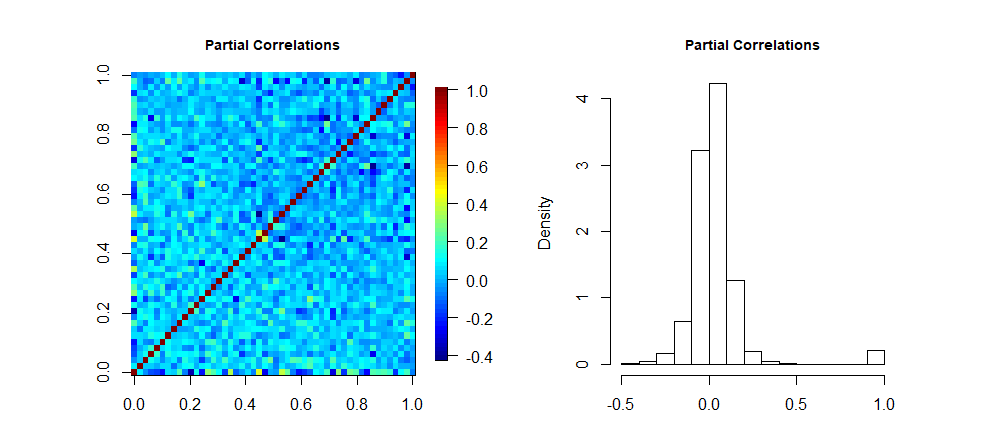}
	\bigskip
	\caption{\textbf{Heatmap and histogram of population partial correlations.} $T=1000$, $p=50$, $q=2$.}
	\bigskip
	\label{fig1a}
\end{figure}
	\begin{figure}[!htb]
	\centering
	\includegraphics[width=.99\linewidth]{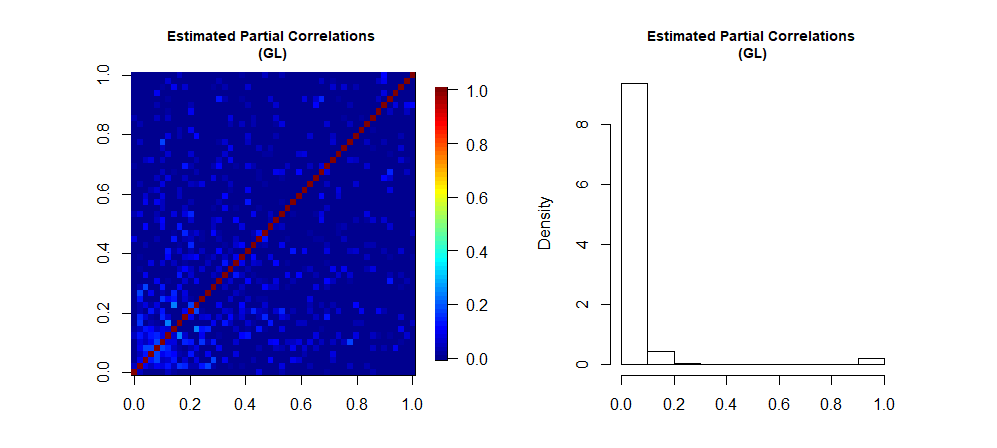}
	\bigskip
	\caption{\textbf{Heatmap and histogram of sample partial correlations estimated using GLASSO with no factors.} $T=1000$, $p=50$, $q=2$, $\hat{q}=0$.}
	\bigskip
	\label{fig1aa}
\end{figure}
	\begin{figure}[!htb]
	\centering
	\includegraphics[width=.99\linewidth]{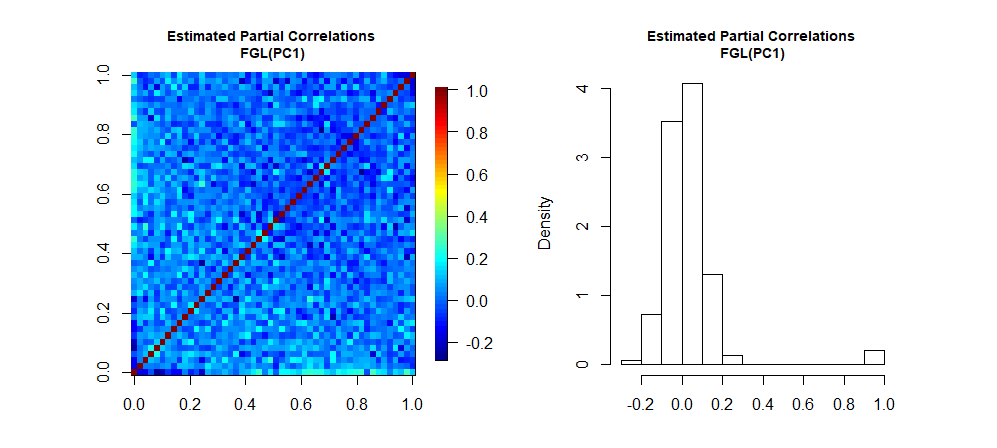}
	\bigskip
	\caption{\textbf{Heatmap and histogram of sample partial correlations estimated using Factor GLASSO with 1 statistical factor.} $T=1000$, $p=50$, $q=2$, $\hat{q}=1$.}
	\bigskip
	\label{fig2a}
\end{figure}
	\begin{figure}[!htb]
	\centering
	\includegraphics[width=.99\linewidth]{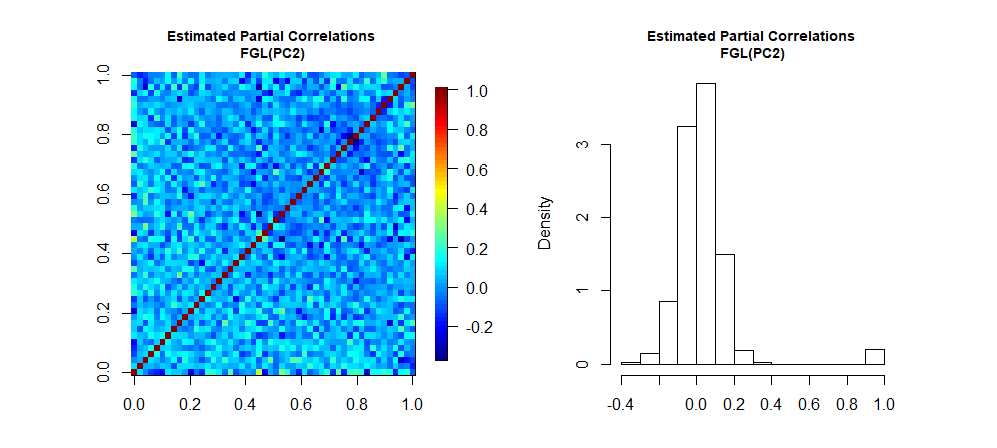}
	\bigskip
	\caption{\textbf{Heatmap and histogram of sample partial correlations estimated using Factor GLASSO with 2 statistical factors.} $T=1000$, $p=50$, $q=2$, $\hat{q}=2$.}
	\bigskip
	\label{fig3a}
\end{figure}
	\begin{figure}[!htb]
	\centering
	\includegraphics[width=.99\linewidth]{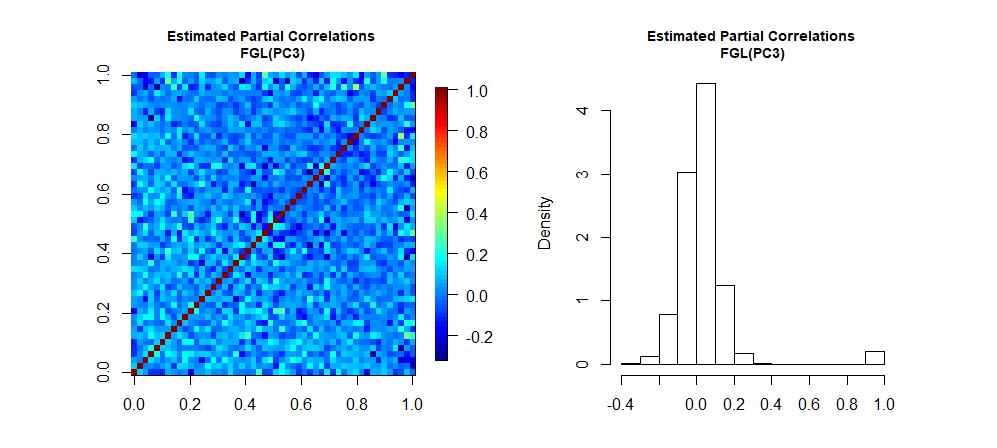}
	\bigskip
	\caption{\textbf{Heatmap and histogram of sample partial correlations estimated using Factor GLASSO with 3 statistical factors.} $T=1000$, $p=50$, $q=2$, $\hat{q}=3$.}
	\bigskip
	\label{fig4a}
\end{figure}
\cleardoublepage

\begin{figure}[!htbp]
	\centering
	\includegraphics[width=0.95\textwidth]{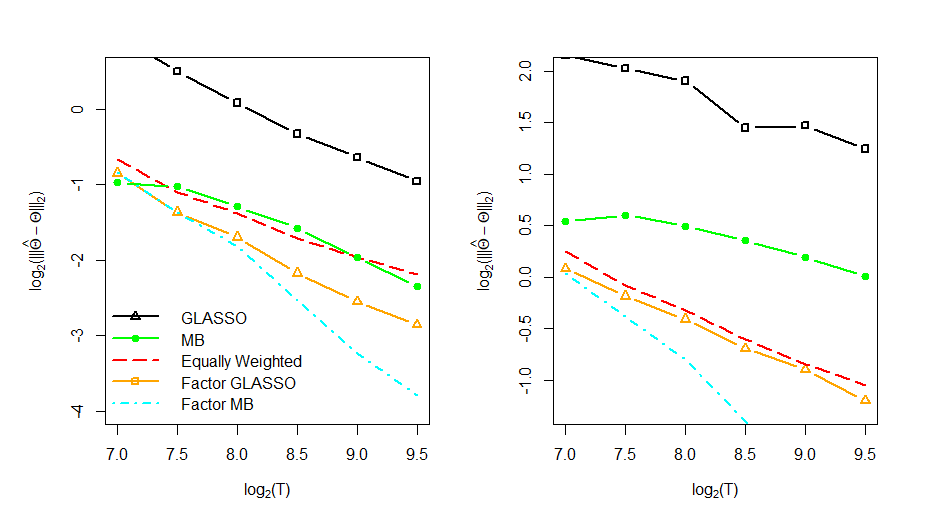}
	\bigskip
	\caption{\textbf{Averaged errors of the estimators of $\bTheta$ on logarithmic scale (base 2). $p = T^{0.85}$, $q = 2(\log(T))^{0.5}$, $s_T = \mathcal{O}(T^{0.05})$.}}
	\label{f18}
\end{figure}
\begin{figure}[!htbp]
	\centering
	\includegraphics[width=0.55\textwidth]{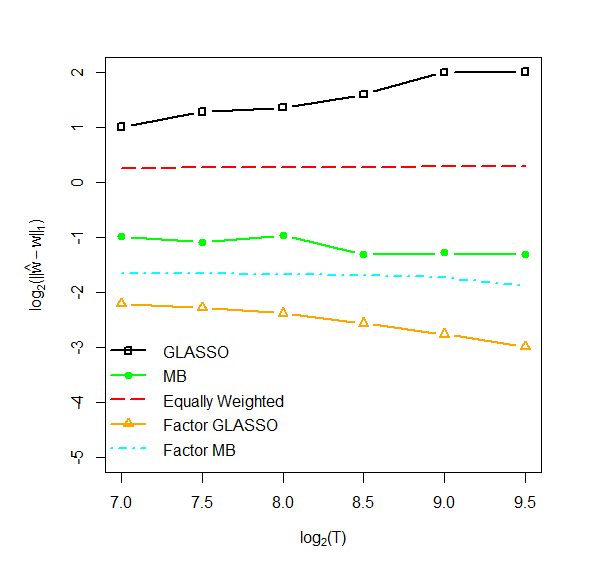}
	\bigskip
	\caption{\textbf{Averaged errors of the estimator of $\bw$ (base 2) on logarithmic scale. $p = T^{0.85}$, $q = 2(\log(T))^{0.5}$, $s_T = \mathcal{O}(T^{0.05})$.}}
	\label{f19}
\end{figure}

\begin{figure}[!htbp]
	\centering
	\includegraphics[width=\textwidth]{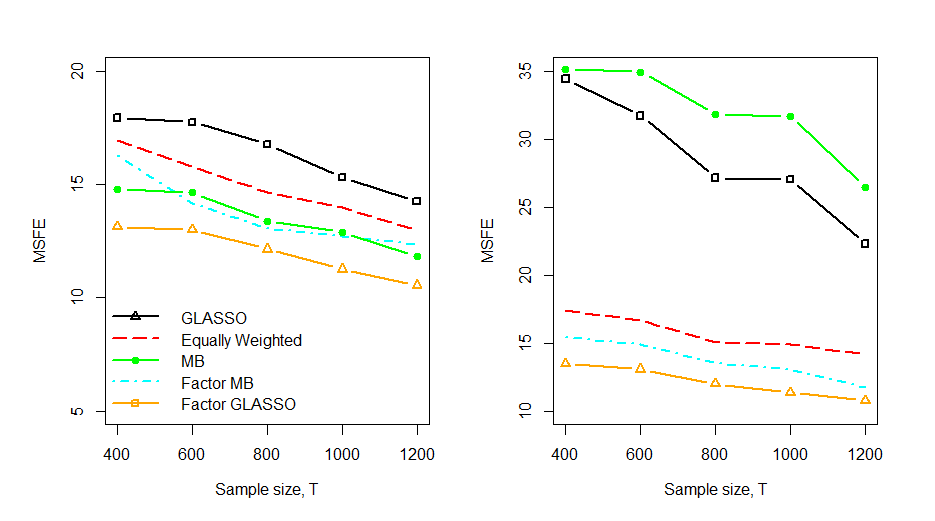}
	\bigskip
\caption{\textbf{Plots of the MSFE over the sample size $\bm{T}$}. $c_1=0$ (left), $c_1=0.75$ (right), $c_2=0.9, \ N=100, \ r=5, \sigma_\xi=1,\ L=7,\ K=2,\ p=24,\ q=5, \ \rho=0.9,\ \phi=0.8$.}
\label{fig2}
\end{figure}
\cleardoublepage

\begin{table}[]
	\phantomsection
	\addcontentsline{toc}{section}{Tables} 
	\centering
	\resizebox{0.7\textwidth}{!}{%
		\begin{tabular}{@{}cccccc@{}}
			\toprule
			& \multicolumn{5}{c}{\textsc{INDPRO}} \\ \midrule
			$h$ & EW & GLASSO & Factor GLASSO & MB & Factor MB \\ \midrule
			1 & 2.77E-04 & 1.51E-04 & \textbf{1.24E-04} & 2.23E-04 & 1.28E-04 \\
			2 & 3.26E-04 & 1.79E-04 & \textbf{5.59E-05} & 1.61E-04 & 1.38E-04 \\
			3 & 1.55E-04 & 9.77E-05 & \textbf{3.81E-05} & 1.17E-04 & 6.54E-05 \\
			4 & 1.18E-04 & 7.60E-05 & \textbf{2.38E-05} & 1.03E-04 & 2.65E-05 \\ \midrule
			& \multicolumn{5}{c}{\textsc{S\&P500}} \\ \midrule
			1 & 1.40E-03 & 1.39E-03 & 1.37E-03 & \textbf{1.34E-03} & 9.57E-03 \\
			2 & 1.71E-03 & 1.44E-03 & \textbf{8.95E-04} & 1.55E-03 & 1.01E-03 \\
			3 & 1.66E-03 & 1.34E-03 & \textbf{3.48E-04} & 1.43E-03 & 6.69E-04 \\
			4 & 1.27E-03 & 1.06E-03 & \textbf{3.95E-04} & 9.55E-04 & 7.91E-04 \\ \midrule			
			& \multicolumn{5}{c}{\textsc{CPI: All Items}} \\ \midrule
			1 & 6.88E-06 & 6.75E-06 & \textbf{5.84E-06} & 6.46E-06 & 8.98E-06\\
			2 & 1.05E-05 & 1.06E-05 & \textbf{8.39E-06} & 9.93E-06 & 9.93E-06 \\
			3 & 1.52E-05 & 1.47E-05 & \textbf{9.36E-06} & 1.56E-05 & 1.34E-05 \\
			4 & 1.63E-05 & 1.63E-05 & \textbf{7.00E-06} & 1.60E-05 & 1.14E-05\\ \midrule		
			& \multicolumn{5}{c}{\textsc{Real Personal Consumption}} \\ \midrule
			1 & 3.05E-05 & \textbf{2.70E-05} & 4.18E-05 & 2.88E-05 & 2.74E-05\\
			2 & 2.65E-04 & 8.52E-05 & 2.79E-05 & 8.11E-05 & \textbf{2.39E-05} \\
			3 & 7.94E-04 & 1.41E-04 & 2.91E-05 & 6.42E-05 & \textbf{2.84E-05} \\
			4 & 8.65E-04 & 7.87E-04 & \textbf{2.61E-05} & 6.42E-05 & 2.63E-05\\ \midrule 						
			& \multicolumn{5}{c}{\textsc{M1 Money Stock}} \\ \midrule
			1 & 5.42E-05 & 5.18E-05 & \textbf{4.99E-05} & 5.40E-05 & 5.47E-05\\
			2 & 5.82E-05 & 1.58E-04 & 7.27E-05 & 5.86E-05 & \textbf{5.40E-05} \\
			3 & 5.97E-05 & 1.56E-04 & 7.44E-05 & 5.96E-05 & \textbf{5.64E-05} \\
			4 & 5.97E-05 & 1.63E-04 & 6.97E-05 & 5.94E-05 & \textbf{5.78E-05}\\ \midrule 															
			& \multicolumn{5}{c}{\textsc{UNRATE}} \\ \midrule
			1 & 0.2531 & 0.0858 & 0.0109 & 0.0557 & \textbf{0.0107} \\
			2 & 0.3758 & 0.1334 & \textbf{0.0066} & 0.0448 & 0.0081 \\
			3 & 0.0743 & 0.0651 & 0.0066 & 0.0532 & \textbf{0.0051} \\
			4 & 2.1999 & 0.6871 & \textbf{0.1578} & 1.0973 & 0.2510 \\ \midrule
			& \multicolumn{5}{c}{\textsc{FEDFUNDS}} \\ \midrule
			1 & 0.0609 & 0.1813 & \textbf{0.0205} & 0.0424 & 0.0448 \\
			2 & 0.1426 & 1.2230 & \textbf{0.0288} & 0.0675 & 0.0416 \\
			3 & 0.2354 & 1.2710 & \textbf{0.0508} & 0.1217 & 0.1038 \\
			4 & 0.3702 & 1.4672 & \textbf{0.0592} & 0.2470 & 0.1962 \\ \bottomrule
		\end{tabular}%
	}
	\caption{\textbf{Prediction of Monthly Macroeconomic Variables.} The numbers are MSFEs with the smallest MSFE in each row in bold font. $h$ indicates the forecast horizon, EW stands for the \enquote{Equal-Weighted} forecast, GLASSO and MB are the models that do not use the factor structure in the forecast errors. Factor GLASSO and Factor MB are our proposed Factor Graphical Models.}
	\label{tab1}
\end{table}
\cleardoublepage
\renewcommand{\appendixpagename}{Supplemental Appendix}
\renewcommand\appendixtocname{Supplemental Appendix}
\begin{appendices} 
	\renewcommand{\thesection}{\Alph{section}}
	\renewcommand{\thesubsection}{\Alph{section}.\arabic{subsection}}
	\renewcommand{\theequation}{\Alph{section}.\arabic{equation}}
	\captionsetup{%
		figurewithin=section,
		tablewithin=section
	}
\section{Additional Simulations}\label{appendixA}
Figures \ref{fig3}-\ref{fi7} show the performance in terms of MSFE for different number of predictors $N$, different values of $c_2$, $\phi$, $\rho$ and $q$: Factor-based models (Factor GLASSO and Factor MB) outperform the equal-weighted forecast combination and the standard GLASSO and nodewise regression without any factor structure. As evidenced from the figures, these findings are robust to the changes in the model parameters. Importantly, Figure \ref{fi7} shows the scenario when the true number of principal components, $r$, is equal to 5, whereas none of the forecasters use PCA for prediction: in this case including at least 2 common components of the forecasting errors reduces MSFE, such that Factor GLASSO and Factor MB outperform EW forecast combination. 
\begin{figure}[hbt!]
	\centering
	\includegraphics[width=0.65\textwidth]{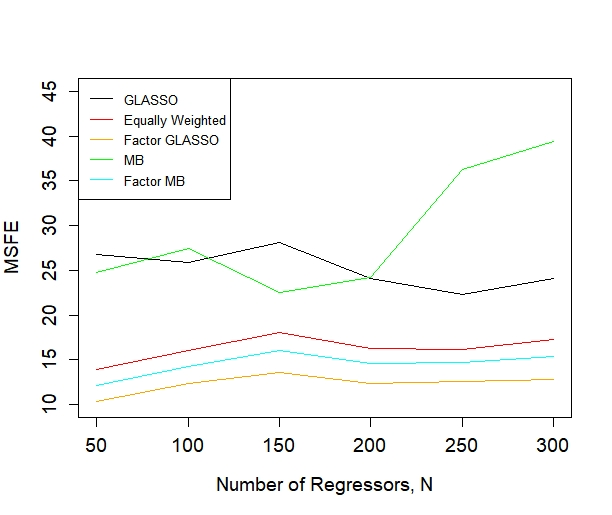}
	\bigskip
	\caption{\textbf{Plots of the MSFE over the number of predictors $\bm{N}$}. $c_1=0.75, \ c_2=0.9, \\ T=800, \ r=5,\ \sigma_\xi=1,\ L=7,\ K=2,\ p=24,\ q=5, \ \rho=0.9,\ \phi=0.8$.}
	\label{fig3}
\end{figure}
\begin{figure}[!htb]
	\centering
	\includegraphics[width=0.65\textwidth]{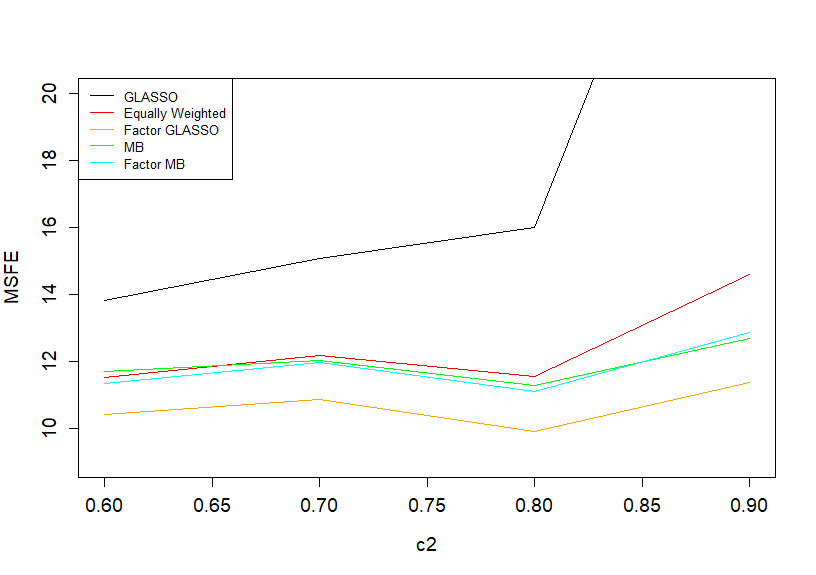}
	\bigskip
	\caption{\textbf{Plots of the MSFE over the values of $\bm{c_2}$.} $c_1=0.75, \ c_2 \in \{0.6, 0.7, 0.8, 0.9\}, \\ T=800, \ N=100$, \ $r=5,\ \sigma_\xi=1,\ L=7,\ K=2,\ p=24,\ q=5, \ \rho=0.9,\ \phi=0.8$.}
	\label{fig4}
\end{figure}
\begin{figure}[hbt!]
	\centering
	\includegraphics[width=0.65\textwidth]{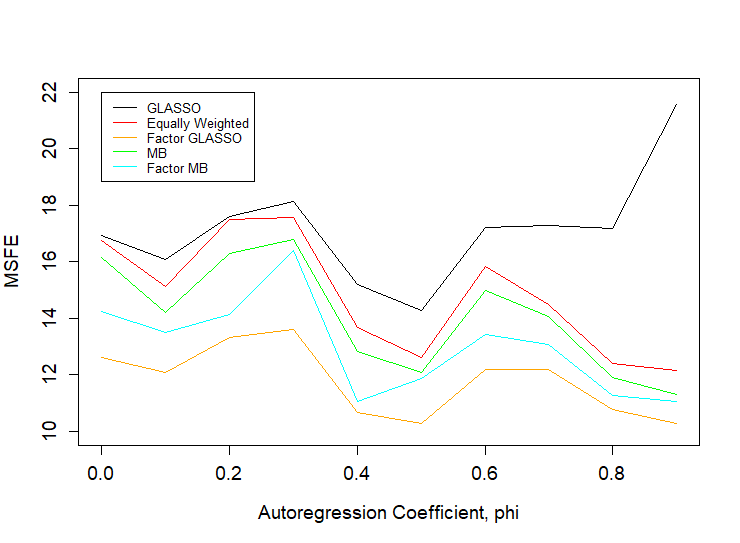}
	\bigskip
	\caption{\textbf{Plots of the MSFE over the values of $\bm{\phi}$.} $c_1=0.75, \ c_2=0.8, \ T=800, \\ N=100, \ r=5,\ \sigma_\xi=1$,\ $L=7,\ K=2,\ p=24,\ q=5, \ \rho=0.9,\ \phi \in \{0, 0.1, \ldots, 0.9\}$.}
	\label{fi5}
\end{figure}
\begin{figure}[hbt!]
	\centering
	\includegraphics[width=0.65\textwidth]{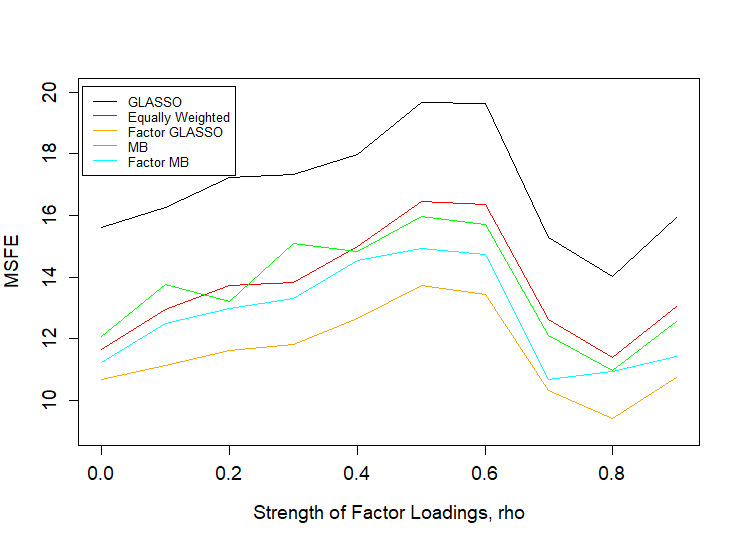}
	\bigskip
	\caption{\textbf{Plots of the MSFE over the values of $\bm{\rho}$.} $c_1=0.75, \ c_2=0.8, \ T=800, \\ N=100, \ r=5,\ \sigma_\xi=1$,\ $L=7,\ K=2,\ p=24,\ q=5, \ \rho \in \{0, 0.1, \ldots, 0.9\},\ \phi=0.7$.}
	\label{fi6}
\end{figure}
\begin{figure}[hbt!]
	\centering
	\includegraphics[width=0.65\textwidth]{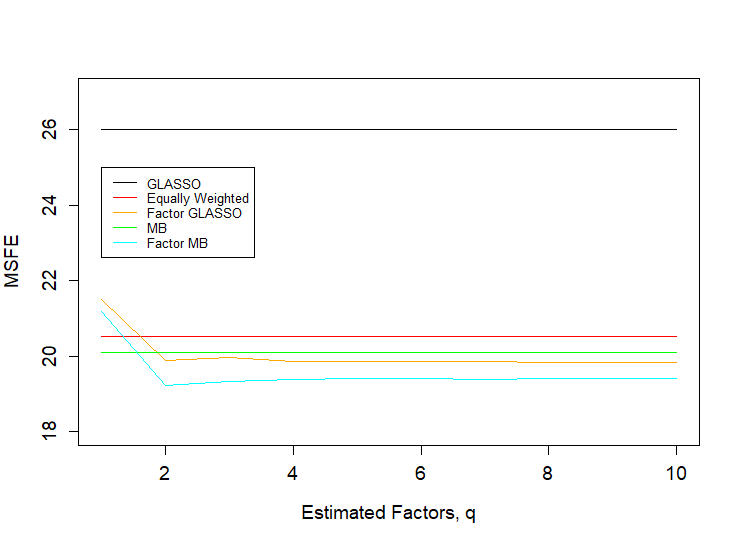}
	\bigskip
	\caption{\textbf{Plots of the MSFE over the values of $\bm{q}$.} $c_1=0.75, \ c_2=0.9, \ T=800, \\ N=100, \ r=5,\ \sigma_\xi=1$,\ $L=12,\ K=0,\ p=13, \ q \in \{0, 1, \ldots, 10\},\ \rho=0.9,\ \phi=0.8$.}
	\label{fi7}
\end{figure}
\end{appendices}
\end{document}